\documentclass[11pt]{article}
\usepackage{amsfonts}
\usepackage{amsmath}
\allowdisplaybreaks[4]
\usepackage{amssymb}
\usepackage{amsthm}
\usepackage{geometry}
\usepackage{indentfirst}
\usepackage{graphicx} 
\usepackage{natbib} \bibpunct{(}{)}{;}{a}{,}{,} 
\usepackage[colorlinks=true,urlcolor=blue,citecolor=blue,linkcolor=blue,bookmarks=true]{hyperref}
\usepackage{appendix}
\usepackage{mathrsfs}
\usepackage{enumerate}
\usepackage{tikz}
\usetikzlibrary{arrows.meta, positioning, bending}
\usepackage{comment}
\usepackage{caption}
\usepackage{subcaption, dsfont} 

\usepackage{algorithm} 
\usepackage{algpseudocode} 
\usepackage{breqn}

\newtheorem{theorem}{Theorem}[section]

\newtheorem{lemma}{Lemma}[section]
\newtheorem{proposition}{Proposition}[section]
\newtheorem{problem}{Problem}[section]
\newtheorem{definition}{Definition}[section]
\newtheorem{example}{Example}[section]

\newtheorem{assumption}{Assumption}[section]

\theoremstyle{remark}

\renewcommand{\cite}{\citet}

\allowdisplaybreaks 

\def\qed{ \hfill \vrule height 6pt width 6pt depth 0pt\newline}

\renewenvironment{proof}{{\noindent\textbf{Proof:}}}{\hfil \qed}
\theoremstyle{definition}

\newcommand{\E}{\mathbb{E}}
\newcommand{\PP}{\mathbb{P}}

\newcommand{\Q}{\mathbb{Q}}

\newcommand{\rmnum}[1]{\romannumeral #1}
\newcommand{\Rmnum}[1]{\uppercase\expandafter{\romannumeral #1\relax}}

\title{\bf Distributionally Robust Insurance under Bregman-Wasserstein Divergence}

\author{Wenjun Jiang\thanks{Department of Mathematics and Statistics, University of Calgary, Calgary, AB, T2N 1N4, Canada. Email: \url{wenjun.jiang@ucalgary.ca}.} \and Qingqing Zhang\thanks{School of Finance, Shanghai University of International Business and Economics, Shanghai 201600, P.R. China. Email: \url{20250036@suibe.edu.cn}.} \and Yiying Zhang\thanks{Department of Mathematics, Southern University of Science and Technology, Shenzhen 518055, Guangdong Province, P.R. China. Email: \url{zhangyy3@sustech.edu.cn.}}}

\date{\today}

\begin{document}

\maketitle

\begin{abstract}
This paper investigates two optimal insurance contracting problems under distributional uncertainty from the perspective of a potential policyholder, utilizing a Bregman-Wasserstein (BW) ball to characterize the ambiguity set of loss distributions. 
The first problem examines an insurance demand model where the policyholder adopts an $\alpha$-maxmin preference with Value-at-Risk (VaR). We derive the optimal indemnity function in closed form and study, both analytically and numerically, how the asymmetry inherent in BW divergence influences the optimal indemnity structure. The second problem employs a robust optimization framework, where the policyholder aims to secure robust insurance indemnity by minimizing the worst-case convex distortion risk measure while adhering to a guaranteed VaR constraint. In this context, we provide explicit characterizations of both the optimal indemnity and the worst-case distribution in closed form through a combined approach using the Lagrangian and relaxation methods. 
To illustrate the practical implications of our theoretical findings, we include a concrete example based on Tail Value-at-Risk (TVaR).

\noindent
    \\[1mm]

\end{abstract}

\section{Introduction}\label{sec:intro}
Insurance serves as an efficient risk‑hedging tool that protects individuals from suffering catastrophic financial losses. An insurance contract is typically characterized by an indemnity function and a premium, where the former specifies the compensation paid to the insurance buyer as a function of the incurred loss, while the latter is the payment made by the insurance buyer to the seller before the contract takes effect. In most practical settings, the premium is determined according to the chosen indemnity function, which simplifies the optimal insurance design problem to the selection of an appropriate indemnity structure. The seminal works \cite{borch1960attempt} and \cite{arrow1974optimal} established the optimality of excess-of-loss indemnity function when insurance buyers aim to minimize the variance of their terminal losses or maximize the expected utility (EU) of their terminal wealth under the expected-value premium principle. Since then, extensive research has been devoted to extending these classical results into more general frameworks. For a comprehensive overview of recent developments on optimal insurance design along various directions, we refer readers to, for example, \cite{albrecher2017reinsurance} and \cite{cai2020optimal}.

A fundamental assumption adopted in most of the literature on optimal insurance contracting is that the insurance buyer possesses perfect knowledge of the loss distribution. However, this assumption is overly idealistic in practice, given the limited or even scarce loss data available to the buyer. The past few years have witnessed growing interest in extending the classical optimal insurance framework, which assumes the loss distribution is known a priori, to settings where the buyer faces ambiguity regarding this distribution, yielding the so-called distributionally robust insurance design problems, which we simply refer to as robust insurance problems in this paper.\footnote{Another major strand of the robust insurance literature is concerned with uncertainty on buyer's preferences; see, for example, \cite{wang2025preference} and \cite{boonen2025robust}.}

The existing literature on robust insurance problems can be categorized by how the uncertainty set for the loss distribution is constructed. Standard statistical approaches, including parametric specification, parameter estimation, and hypothesis testing, typically produce a finite number of plausible loss distributions. In such a setting, \cite{asimit2017robust, asimit2019optimal} studied the robust insurance contract that minimizes the insurance buyer's terminal risk exposure, where the buyer is assumed to be extremely ambiguity averse and only considers the worst-case distribution among the finite choices. \cite{jiang2020optimal} studied a similar problem under the smooth ambiguity model, which shows that ambiguity aversion results in a distorted probability distribution over the set of possible models with a bias in favor of the model that yields a larger risk.

The uncertainty set that contains infinitely many plausible loss distributions is often constructed by using the moment-based or distance-based approaches. In the former line of research, \cite{liu2022distributionally} and \cite{xie2023distributionally} are the two representative works that studied the robust insurance problem for excess-of-loss indemnity functions under complete knowledge of the first two moments of the loss, where the insurance buyer is assumed to hold Value-at-Risk (VaR), Tail Value-at-Risk (TVaR) or expectile preference. In the latter line of research, a substantial body of literature focuses on Wasserstein balls centered at a benchmark distribution. The adoption of the Wasserstein distance originates from optimal transport theory \citep{villani2009optimal} and has been shown to offer several advantages over alternative distance measures \citep{gao2023distributionally}. For the optimal insurance problems, \cite{cai2024worst, cai2026worst} studied the worst-case risk assessment of limited- and stop-loss transform of variables when the uncertainty set is a Wasserstein ball with and without information of the first two moments of the variables. \cite{boonen2024robust, boonen2025distributionally} and \cite{jiang2025optimal} studied the robust insurance problem without imposing a parametric form on the indemnity function. In their work, the buyer minimizes a risk evaluated by a convex distortion risk measure or VaR, with the uncertainty set defined as either an $L^k$ or a $k$-Wasserstein ball. We also refer interested readers to, for example, \cite{birghila2021} for numerical approaches to solving robust insurance problems under max-min EU models, and to \cite{liu2022distributionally} for the study of inf-convolution problems under robust coherent distortion risk measures and VaR.

This paper complements the existing studies of robust insurance problems by considering a Bregman--Wasserstein (BW) ball. The BW divergence reduces to the squared $2$-Wasserstein distance when the Bregman generator is quadratic. More generally, unlike the symmetric Wasserstein distance, it allows positive and negative deviations from the benchmark distribution to incur different transport costs. In insurance practice, claim inflation, reporting delays, and changes in portfolio composition may cause the benchmark distribution to underestimate or overestimate future losses. These two types of errors have different consequences: underestimation may leave the policyholder underinsured, whereas overestimation may lead to excessive coverage and unnecessary premium payments. The Bregman generator allows upward and downward deviations from the benchmark distribution to be penalized differently, thereby providing a flexible way to capture such asymmetric distributional uncertainty. The BW ambiguity set has also been employed by \cite{pesenti2025optimal} in a robust portfolio choice problem. We investigate two problems in this paper. The first problem considers an ambiguity-averse decision maker (DM) whose preference is modeled by a VaR-based $\alpha$-maxmin criterion. We derive explicit expressions for the worst- and best-case VaR of the loss when its distribution belongs to a BW ball, as well as for the optimal indemnity function. We find that for a more ambiguity-averse DM, increased sensitivity to positive deviations from the benchmark quantile function reduces insurance demand. In contrast, greater sensitivity to negative deviations does not consistently influence demand. This result can be reversed when the DM is ambiguity-seeking. 

In the second problem, we consider the robust optimal insurance problem when the DM adopts a general convex distortion risk measure and the distributional  uncertainty set is characterized by a BW ball under a constraint placed on the worst-case VaR of the DM's terminal loss, which is referred to as the guaranteed VaR performance model. This model is conceptually aligned with the robust mean-variance framework of \cite{blanchet2022distributionally}, in which the worst-case portfolio variance is minimized subject to a constraint on the worst-case expected return. We derive closed-form expressions for the optimal indemnity function and the worst-case survival function by combining the Lagrangian method with \textit{relaxation} approach (also called as the \textit{modification} method; see \cite{tan2020optimal,meng2022optimal,meng2024optimal}). The BW divergence governs the cost of distributional deviations, while the guaranteed VaR constraint induces a discontinuity in the worst-case survival function obtained from pointwise optimization, thereby necessitating the use of relaxation arguments. Our results indicate that as the guaranteed VaR constraint becomes more stringent, the DM is increasingly inclined to exaggerate tail losses while adhering to the benchmark distribution for small and moderate losses.


The rest of this paper is structured as follows. Section \ref{sec:preliminaries_and_setup} briefly reviews the distortion risk measure and the BW divergence, and sets up the two main problems that will be addressed in this paper. Sections \ref{sec:sol-prob1} and \ref{sec:sol-prob2} present the solutions to the two main problems with economic insights provided to show how BW ball and guaranteed performance affect the DM's demand for insurance. In particular, Section \ref{sec:concrete} presents a concrete example to show further implications of our main results. Section \ref{sec:conclusion} concludes the paper and gives several directions for future research. All the proofs are delegated to the appendix.

\section{Preliminaries and problem formulation}\label{sec:preliminaries_and_setup}
Let $(\Omega,\mathcal{F})$ be a measurable space, where $\Omega$ is the sample space and $\mathcal{F}$ is a $\sigma$-algebra that contains all the subsets of $\Omega$. Let $\PP$ be a probability measure and $Z$ be a random variable defined on $(\Omega,\mathcal{F})$, we denote by $F$ the cumulative distribution function (CDF) of $F$ under $\PP$, i.e. $F(z)=\PP(Z\le z)$, and by $S$ the survival function of $Z$, i.e. $S(z)=1-F(z)$. When the CDF of $Z$ is a particular $F$, we use the notation $Z\sim F$. Throughout the paper, we also adopt the following notations: $(x)_+=\max\{x,0\}$, $x\wedge y=\min\{x,y\}$ and $x\vee y=\max\{x,y\}$. The indicator function $\mathds{1}_A(x)$ is equal to $1$ if $x\in A$ and $0$ otherwise.

Before presenting the main problems to be studied in this paper, we give a brief review of the distortion risk measure and Bregman-Wasserstein divergence in the following two subsections.

\subsection{Distortion risk measures}
The distortion risk measure stems from Yaari's dual theory of choice \citep{yaari1987dual} and evaluates a risk based on its distorted expectation. For a random variable $Z$ built on $(\Omega,\mathcal{F},\PP)$, the distortion risk measure of $Z$ is defined as 
\begin{equation}
    \rho_g^\PP(Z)=\int_0^\infty g(\PP(Z>z))dz+\int_{-\infty}^0[g(\PP(Z>z))-1]dz,
\end{equation}
where $g$ is called the \emph{distortion function}, which is increasing over its domain $[0,1]$ and satisfies $g(0)=0$ and $g(1)=1$. The popularity of the distortion risk measure is due to its nice mathematical properties, such as monotonicity, positive homogeneity, translation invariance, and comonotonic additivity. The predominant risk measures in finance and insurance, VaR and TVaR, are both special cases of the distortion risk measure. For a random variable $Z$, its VaR at the confidence level $\alpha\in(0,1)$ is defined as its left quantile at $\alpha$: $\mathrm{VaR}_\alpha(Z)=\inf\left\{z:\ \PP(Z\le z)\ge\alpha\right\}$, and its TVaR at the same confidence level is defined as $\mathrm{TVaR}_\alpha(Z)=\frac{1}{1-\alpha}\int_\alpha^1\mathrm{VaR}_s(Z)\,ds$. The distortion functions of $\mathrm{VaR}_\alpha(\cdot)$ and $\mathrm{TVaR}_\alpha(\cdot)$ are $g_{V,\alpha}(t)=\mathds{1}_{(1-\alpha,1]}(t)$ and $g_{T,\alpha}(t)=\frac{t}{1-\alpha}\wedge 1$ respectively.

If $g$ is concave, the distortion risk measure is convex. The convex distortion risk measure is coherent in the sense of \cite{artzner1999coherent} and preserves the increasing convex order between random variables. In the rest of this paper, we will particularly focus on VaR and convex distortion risk measures, which are common choices for modeling the preferences of decision makers. For more examples and applications of distortion risk measures, we refer the interested readers to \cite{denuit2006actuarial}.

\subsection{Bregman--Wasserstein divergence}
\label{sec:BW_uncertainty} 
The BW divergence is the optimal transport cost when the cost function in the Monge-Kantorovich problem is replaced by the Bregman divergence. Although the literature is now rich with introduction of the BW divergence, we provide a brief review in this section to make the paper self-contained. We recall the Bregman divergence in the following.

\begin{definition}[Bregman divergence]
Let $\varphi : \mathbb{R} \to \mathbb{R}$ be a strictly convex and continuously differentiable function, 
called a \emph{Bregman generator}. The Bregman divergence with generator $\varphi$ 
is defined as
\begin{equation*}
    B_{\varphi}(x,y)
    := \varphi(x) - \varphi(y) - \varphi'(y)(x - y), 
    \quad x, y \in \mathbb{R},
\end{equation*}
where $\varphi'(x) := \tfrac{d}{dx}\varphi(x)$ denotes the derivative of $\varphi$.
\end{definition}

The Bregman divergence includes some well-known distances and divergences as special cases. For example, if letting $\varphi(x)=x^2$, the Bregman divergence becomes the squared distance $B_\varphi(x,y)=(x-y)^2$ (or squared Euclidean distance for higher-dimensional space). When focusing on the discrete distributions, e.g., $\boldsymbol{x}=(x,\dots,x_d)\in[0,1]^d$ and $\boldsymbol{y}=(y_1,\dots,y_d)\in[0,1]^d$, if letting $\varphi(\boldsymbol{x})=\sum_{i=1}^dx_i\ln x_i$, the Bregman divergence becomes the Kullback-Leibler divergence. We refer the interested readers to Table 1 of \cite{banerjee2005clustering} for more special cases of the Bregman divergence. In the statistical field, the most often used statistical functional is the mean. It can be easily verified that the mean is elicitable by the Bregman divergence, and \cite{gneiting2011making} shows that the scoring function being a Bregman divergence is not only sufficient but also necessary for doing so. 

We are now ready to give the definition of the BW divergence that is built on the space of CDFs, which is simply the minimum cost resulting from solving the following Monge-Kantorovich optimal transport problem 
\begin{equation*}
    \mathscr{B}_{\varphi}[F_1,F_2] := 
    \inf_{X\sim F_1, Y\sim F_2}\ \E[B_\varphi(X,Y)].
\end{equation*}
It is shown by \cite{pensenti2024optimal} that the infimum is attained by the comonotonic coupling of $X$ and $Y$, i.e., $(X,Y)\overset{d}{=}(F_1^{-1}(U), F_2^{-1}(U))$ with $U$ being a standard uniform random variable. Consequently, the BW divergence from $F_1$ to $F_2$ can be written in terms of their quantile functions
\begin{equation*}
\begin{aligned}
    \mathscr{B}_{\varphi}[F_1,F_2]
    &=\inf_{X\sim F_1,Y\sim F_2}\ \E[B_\varphi(X,Y)] \\
    &=\E[B_\varphi(F_1^{-1}(U),F_2^{-1}(U))] \\
    &=\int_0^1 B_{\varphi}\big(F_1^{-1}(t), F_2^{-1}(t)\big)\,dt \nonumber \\
    &= \int_0^1 \big(
        \varphi(F_1^{-1}(t))
        - \varphi(F_2^{-1}(t))
        - \varphi'(F_2^{-1}(t))
        \big(F_1^{-1}(t) - F_2^{-1}(t)\big)
        \big)\,dt.
    \label{eq:bw_divergence}
\end{aligned}
\end{equation*}
Thus, the BW divergence generalizes the squared $2$-Wasserstein distance, which is recovered when $\varphi(x)=x^2$. For additional geometric properties and more detailed interpretations of the BW divergence, we refer the reader to \cite{kainth2025bregman}. In addition to the classical quantile representation, $\mathscr{B}_\varphi[F_1,F_2]$ admits another survival-function-based representation under mild conditions, which is summarized in the following lemma.
\begin{lemma}\label{lem:survival-rep-BW}
    If the support of both $F_1$ and $F_2$ is $[0,M]$ and $\sup_{x\in [0,M]}|\varphi'(x)| < \infty$, where $0<M<\infty$, then
\begin{equation}\label{eq:bw_divergence2}
\begin{aligned}
    \mathscr{B}_{\varphi}[F_1,F_2]
    =&\int_0^M [\varphi'(x)-\varphi'(0)]S_1(x)\,dx+\int_0^M \varphi''(x)xS_2(x)\,dx\\
    &-\int_0^M\int_0^{M}\varphi''(y)S_2(y)\wedge S_1(x)\,dy\,dx.
\end{aligned}
\end{equation}
\end{lemma}
\noindent Typical examples of $\varphi$ that fulfill the condition of Lemma \ref{lem:survival-rep-BW} are, for instance, $\varphi(x)=x^p$ for $p>1$ and $\varphi(x)=(x+a)\ln(x+a)$ for $a>0$.

Notably, the BW divergence is non-symmetric and penalizes the positive and negative deviations from the benchmark distribution differently in the sense that $\mathscr{B}_\varphi[F_1,F_2]\neq \mathscr{B}_\varphi[F_2,F_1]$ and $\mathscr{B}_\varphi[F_2+\epsilon \delta, F_2]\neq\mathscr{B}_\varphi[F_2-\epsilon \delta, F_2]$,\footnote{Here $\epsilon$ and $\delta$ are chosen such that $F_2\pm\epsilon\delta$ are still CDFs.} as elaborated by \cite{pesenti2025optimal}. These properties, which are not possessed by symmetric distance metrics in probability space, are of great interest in decision-making problems. 

\subsection{Problem formulation}\label{sec:prob_formulation}
In this section, we set up the problems that will be studied in this paper. We focus on a one-period economy. Let $X$ now be a non-negative bounded loss faced by a potential policyholder (hereafter referred to as decision maker, or DM for short), whose support is a subset of $[0,M]$ with $0<M<\infty$. The DM is interested in purchasing insurance $(I,\pi(\cdot))$ from an insurer. We denote by $\PP$ and $\Q$ the probability measures reflecting the decision maker’s and the insurer’s subjective perceptions, respectively. For the insurance premium, the classic expected-value premium principle is adopted, i.e. $\pi(I(X))=(1+\theta)\E^\Q[I(X)]$ with $\theta>0$ being the safety loading factor, where $\E^\Q[\cdot]$ denotes the expectation evaluated under the probability measure $\Q$. 

Aligning with the literature, we assume that the indemnity function is chosen from the following class
\begin{equation}
    \mathcal{I}=\left\{I:[0,M]\mapsto[0,M]:\ I(0)=0,\ 0\le I(y)-I(x)\le y-x,\ \forall\ x,y\in[0,M]\right\}.
\end{equation}
The 1-Lipschitz continuity imposed by the class $\mathcal{I}$ is called the \emph{no-sabotage} condition as per \cite{huberman1983optimal} and \cite{carlier2003pareto}, under which the so-called \emph{ex post} moral hazard issues can be alleviated. The class $\mathcal{I}$ is rich and encompasses many well-known indemnity functions, such as the stop-loss, proportional and limited stop-loss functions. 

With the insurance contract, the DM's terminal loss becomes $X-I(X)+\pi(I(X))$. We assume that the DM is ambiguous about the underlying distribution of $X$ and employs the BW divergence to depict the uncertainty. In the following, we denote the uncertainty set of $X$'s distribution by
\begin{equation*}
   \mathcal{B}(F_0,\epsilon):=\left\{F:F\ \text{is the CDF of $X$, and}\ \mathscr{B}_\varphi[F,F_0]\le\epsilon\right\} 
\end{equation*}
with $\epsilon>0$, where $F_0$ is the benchmark distribution of $X$ under $\PP$. The asymmetric nature of the BW divergence allows the DM to assign different penalties to positive and negative deviations from the benchmark distribution. The choice of the Bregman generator $\varphi$ can therefore reflect the relative plausibility of deviations in the two directions, while the radius $\epsilon$ controls the overall degree of distributional uncertainty. Example~\ref{example:alpha-maxmin-var} illustrates how such directional sensitivity can be incorporated through the choice of the generator.

We consider two insurance decision problems under distributional uncertainty. The first problem allows the DM to exhibit different attitudes toward distributional ambiguity. Following the $\alpha$-maxmin criterion, the DM assigns weights $\kappa$ and $1-\kappa$ to the worst- and best-case VaR, respectively. 
The notations $\mathrm{VaR}_\alpha^{F_i}$ and $\rho_g^{F_i}$ denote the corresponding risk measures evaluated under the CDF $F_i$.
\begin{problem}[$\alpha$-maxmin VaR-based model]\label{Prob:main1}
    Let $\alpha\in(0,1)$ and $\kappa\in[0,1]$, solve 
    \begin{equation*}
     \inf_{I\in\mathcal{I}}\ \left\{
     \kappa\sup_{F_1\in\mathcal{B}(F_0,\epsilon)}
     \mathrm{VaR}^{F_1}_\alpha(X-I(X)+\pi(I(X)))
     +(1-\kappa)\inf_{F_2\in\mathcal{B}(F_0,\epsilon)}
     \mathrm{VaR}^{F_2}_\alpha(X-I(X)+\pi(I(X)))
     \right\}. 
    \end{equation*}
\end{problem}
\noindent Here, a larger $\kappa$ corresponds to greater ambiguity aversion, whereas a smaller $\kappa$ places more emphasis on favorable distributional scenarios.

In the second problem, the DM minimizes the worst-case convex distortion risk measure of the terminal loss, while requiring its worst-case VaR to remain below a prescribed level. This is motivated by the fact that minimizing an overall risk measure does not necessarily ensure that risk-management and regulatory requirements are satisfied. We therefore consider the following guaranteed VaR performance model.
\begin{problem}[Guaranteed VaR performance model]\label{Prob:main2}
    Let $\rho_g(\cdot)$ be a convex distortion risk measure and 
    $0<A<\sup_{F\in\mathcal{B}(F_0,\epsilon)}\mathrm{VaR}^F_\alpha(X)$, solve
    \begin{equation*}
    \left\{
    \begin{aligned}
        \inf_{I\in\mathcal{I}}&\ 
        \sup_{F_1\in\mathcal{B}(F_0,\epsilon)}\ 
        \rho_g^{F_1}(X-I(X)+\pi(I(X))), \\
        \mbox{s.t.}&\ 
        \sup_{F_2\in\mathcal{B}(F_0,\epsilon)}
        \mathrm{VaR}^{F_2}_\alpha(X-I(X)+\pi(I(X)))\le A.
    \end{aligned}
    \right.
    \end{equation*}
\end{problem}
\noindent Here, $A$ specifies the maximum acceptable worst-case VaR, with a smaller $A$ representing a more stringent performance requirement.

\section{The solution to Problem \ref{Prob:main1}}\label{sec:sol-prob1}
Following the idea of \cite{jiang2025optimal}, the inner problem of Problem \ref{Prob:main1} can be simplified as
\begin{align}
    &\kappa\sup_{F_1\in\mathcal{B}(F_0,\epsilon)}\mathrm{VaR}^{F_1}_\alpha(X-I(X)+\pi(I(X)))+(1-\kappa)\inf_{F_2\in\mathcal{B}(F_0,\epsilon)}\mathrm{VaR}^{F_2}_\alpha(X-I(X)+\pi(I(X))) \nonumber \\
    =\ &\kappa R\left(\sup_{F_1\in\mathcal{B}(F_0,\epsilon)}\mathrm{VaR}_\alpha^{F_1}(X)\right)+(1-\kappa)R\left(\inf_{F_2\in\mathcal{B}(F_0,\epsilon)}\mathrm{VaR}_\alpha^{F_2}(X)\right)+\pi(I(X)), \label{Prob:main1-v2}
\end{align}
where $R(x):=x-I(x)$ denotes the DM's retained loss function\footnote{Clearly, for any $I\in\mathcal{I}$, it must hold $R\in\mathcal{I}$.}.  Hence, solving Problem \ref{Prob:main1} requires first the evaluation of $\sup_{F_1\in\mathcal{B}(F_0,\epsilon)}\mathrm{VaR}_\alpha^{F_1}(X)$ and $\inf_{F_2\in\mathcal{B}(F_0,\epsilon)}\mathrm{VaR}_\alpha^{F_2}(X)$. The worst- and best-case VaR of $X$ are summarized in the following theorems.

\begin{proposition}\label{th4-1}
Let $\epsilon\in(0,\infty)$. If $\epsilon > \int_{\alpha}^{1} B_{\varphi}(M,F_0^{-1}(t))\,dt$, then $\sup_{F \in \mathcal{B}(F_0,\epsilon)} \mathrm{VaR}^F_{\alpha}(X)=M$. If 
$\epsilon \leq \int_{\alpha}^{1} B_{\varphi}(M,F_0^{-1}(t))\,dt$, then 
\begin{equation}\label{eqsup}
\sup_{F \in \mathcal{B}(F_0,\epsilon)} \mathrm{VaR}^F_{\alpha}(X)=\overline{V_\alpha}:= \inf \left\{ D \in \left( F_0^{-1}(\alpha), M \right] :
\int_{\alpha}^{F_0(D)} B_{\varphi}(D,F_0^{-1}(t))\, dt \geq \epsilon \right\}. 
\end{equation}
Furthermore, the supremum can never be attained in the sense that there does not exist any $F^{*} \in \mathcal{B}(F_0,\epsilon)$ such that $\mathrm{VaR}_{\alpha}^{F^{*}}(X)
= \overline{V_\alpha}$.
\end{proposition}

\begin{proposition}\label{th4-2}
For any $\epsilon\in(0,\infty)$, 
\begin{equation}\label{eqinf}
\inf_{F \in \mathcal{B}(F_0,\epsilon)} \mathrm{VaR}^F_{\alpha}(X)= \underline{V_\alpha}:= \inf \left\{ D \in \left[0, F_0^{-1}(\alpha) \right) :
\int_{F_0(D)}^{\alpha} B_{\varphi}(D,F_0^{-1}(t))\, dt \le \epsilon \right\}. 
\end{equation}
Furthermore, the infimum can be attained by some $F^{*} \in \mathcal{B}(F_0,\epsilon)$ with $\mathrm{VaR}_{\alpha}^{F^{*}}(X)
= \underline{V_\alpha}$.
\end{proposition}

Although there does not exist any $F^{*} \in \mathcal{B}(F_0,\epsilon)$ such that $\mathrm{VaR}_{\alpha}^{F^{*}}(X)
= \overline{V_\alpha}$, we can find a sequence of CDFs $\{F_n\}\in \mathcal{B}(F_0,\epsilon)$ such that $\lim_{n\to \infty}\mathrm{VaR}_{\alpha}^{F_n}(X)=\overline{V_\alpha}$. More specifically, define 
\begin{equation}\label{eqhatFn}
F_n^{-1}(t)=
\begin{cases}
\overline{V_\alpha}-\delta_{n}, & t\in(\alpha-\xi_{n}, F_0(\overline{V_\alpha}-\delta_{n})], \\
F_0^{-1}(t), & t\notin(\alpha-\xi_{n}, F_0(\overline{V_\alpha}-\delta_{n})],
\end{cases}    
\end{equation}
where $\delta_0\in (0,F_0^{-1}(\alpha)]$, $\{\delta_n\}$ is a decreasing sequence with $\lim_{n\to \infty}\delta_n=0$, and for each $n$, there exists $\xi_n\in(0,\alpha)$ such that
\begin{align*} &\int_{0}^{1}B_{\varphi}(F_n^{-1}(t),F_0^{-1}(t))\ dt=\int_{\alpha-\xi_{n}}^{F_0(\overline{V_\alpha}-\delta_{n})}B_{\varphi}(\overline{V_\alpha}-\delta_{n},F_0^{-1}(t))\ dt \le\epsilon. \end{align*}
Then, it is straightforward to verify that $\lim_{n\to \infty}F_n^{-1}(\alpha)=\lim_{n\to \infty}(\overline{V_\alpha}-\delta_{n})=\overline{V_\alpha}.$

With Propositions \ref{th4-1} and \ref{th4-2}, the objective function of Problem \ref{Prob:main1} reduces to
\begin{align*}
    &\kappa R(\overline{V_\alpha})+(1-\kappa)R(\underline{V_\alpha})+\pi(I(X)) \nonumber \\
    =\ &\kappa\overline{V_\alpha}+(1-\kappa)\underline{V_\alpha}+(1+\theta)\E^\Q[I(X)]-\kappa I(\overline{V_\alpha})-(1-\kappa)I(\underline{V_\alpha}) \\
    =\ &\kappa\overline{V_\alpha}+(1-\kappa)\underline{V_\alpha}+\int_0^M\left\{(1+\theta)S_\Q(x)-\kappa\mathds{1}_{[0,\overline{V_\alpha}]}(x)-(1-\kappa)\mathds{1}_{[0,\underline{V_\alpha}]}(x)\right\}dI(x).
\end{align*}
The solution to Problem \ref{Prob:main1} is presented in the following theorem, for which the proof is straightforward and thus omitted.

\begin{theorem}\label{th3.3}
    Let $H(t)=(1+\theta)S_\Q(t)-\kappa\mathds{1}_{[0,\overline{V_\alpha}]}(t)-(1-\kappa)\mathds{1}_{[0,\underline{V_\alpha}]}(t)$. The optimal indemnity function that solves Problem \ref{Prob:main1} is given by $I^*(x)=\int_0^x{I^*}'(t)\,dt$, where
    \begin{equation}
        {I^*}'(t)=\mathds{1}_{\{t: H(t)<0\}}(t)+\eta(t)\mathds{1}_{\{t: H(t)=0\}}(t)
    \end{equation}
    with $\eta(t)$ being any $[0,1]$-valued Lebesgue measurable function.
\end{theorem}

The function $H(t)$ can be easily re-written as
\begin{equation*}
    H(t) = (1+\theta) S_\Q(x) - \mathds{1}_{[0,\underline{V_\alpha}]}(t) - \kappa\mathds{1}_{(\underline{V_\alpha}, \overline{V_\alpha}]}(t).
\end{equation*}
To make a comparison with existing literature results, we take $\eta(t)=1$ when $H(t)=0$. Define 
\begin{equation*}
    d_1^*=\inf\left\{x\in[0,M]: S_\Q(x)\le \frac{1}{1+\theta}\right\},\quad
    d_2^*=\inf\left\{x\in[0,M]: S_\Q(x)\le \frac{\kappa}{1+\theta}\right\}, 
\end{equation*}
we have
\begin{equation}\label{eq:explicit_I_var}
    I^*(x)=\left\{
    \begin{aligned}
        &(x\wedge \underline{V_\alpha}-d_1^*)_+,&\quad& \text{if}\ d_1^*<\underline{V_\alpha}<\overline{V_\alpha}\le d_2^*, \\
        &(x\wedge \underline{V_\alpha}-d_1^*)_++(x\wedge \overline{V_\alpha}-d_2^*)_+,&\quad& \text{if}\ d_1^*<\underline{V_\alpha}\leq d_2^*<\overline{V_\alpha}, \\
        &(x\wedge \overline{V_\alpha}-d_2^*)_+,&\quad& \text{if}\ \underline{V_\alpha}\leq d_1^*\leq d_2^*<\overline{V_\alpha}, \\
        &(x\wedge \overline{V_\alpha}-d_1^*)_+,&\quad& \text{if}\ d_1^*\leq d_2^*<\underline{V_\alpha}<\overline{V_\alpha},\\
        &0,&\quad&\ \text{otherwise}.
    \end{aligned}
    \right.
\end{equation}
Unlike the classical VaR-based model (without model uncertainty), where a limited excess-of-loss indemnity is optimal (cf. \cite{cheung2010optimal} and \cite{chi2011optimal}), the optimal indemnity for Problem \ref{Prob:main1} can involve a more layered structure. This illustrates how the decision maker’s ambiguity aversion influences insurance demand. Proposition \ref{th4-1} (resp. Proposition \ref{th4-2}) shows that computing the worst-case (resp. best-case) VaR only requires examining the CDFs whose quantile functions deviate positively (resp. negatively) from the benchmark. Consequently, the explicit indemnity form in \eqref{eq:explicit_I_var} reveals the effect of the asymmetry in the BW divergence on insurance demand. We conclude this section with an illustrative example.

\begin{example}\label{example:alpha-maxmin-var}
    We assume that under the benchmark distribution the loss follows a truncated exponential distribution with mean $1$ (million) and truncation point $100$ (million). The confidence level for VaR is set to be $\alpha=0.95$. For convenience, let $q_\alpha=\mathrm{VaR}^{F_0}_\alpha(X)$. The safety loading factor and the ambiguity aversion level are assumed to be $\theta=0.5$ and $\kappa=0.9$ respectively. We consider the following $\varphi$ with $k$ signifying the DM's sensitivity of the positive deviation from the benchmark quantile function:
     \begin{equation*}
    \varphi(x;k)=\left\{
    \begin{aligned}
        &x^2,&\quad &x\in[0,q_\alpha), \\
        &q_\alpha^2+2q_\alpha(x-q_\alpha)+k(x-q_\alpha)^2,&\quad &x\in[q_\alpha,M].
    \end{aligned}
    \right.
     \end{equation*}
When $k=1$, $\varphi(x;k)=x^2$, so the BW divergence reduces to the squared $2$-Wasserstein distance. Note that $\varphi(x;k_2)-\varphi(x;k_1)$ is convex if $k_2>k_1$, thus $B_{\varphi(\cdot;k_2)}(D,F_0^{-1}(t))>B_{\varphi(\cdot;k_1)}(D,F_0^{-1}(t))$ for all $t\in(\alpha,F_0(D))$ when $D>F_0^{-1}(\alpha)$. Apparently, a larger $k$ penalizes more the CDFs whose quantile functions deviate positively from the benchmark quantile function (or shows that the DM is more confident about the benchmark distribution). Under a given $\epsilon$ (which is 0.5 in this example), it is not a surprise to find that $\overline{V_\alpha}$ decreases with $k$. Under our setting, it can be easily calculated that $d_1^*\approx 0.406$, $d_2^*\approx 0.511$ and $\underline{V_\alpha}\approx 0.563$, where the $\underline{V_\alpha}$ remains the same across different values of $k$. Hence, the optimal indemnity function is given by
$I^*(x)=(x\wedge\overline{V_\alpha}-d_1^*)_+$.
Compared with the symmetric case $k=1$, increasing $k$ to $4$ reduces the upper limit of coverage from approximately $6.596$ to $4.894$ million, while leaving the deductible $d_1^*$ unchanged. Thus, in this example, greater sensitivity to positive
deviations from the benchmark reduces the DM's insurance demand by limiting coverage for larger losses. The corresponding optimal indemnity functions are shown in the right panel of
Figure~\ref{fig:varphi_and_indemnity}.


In contrast, if the DM penalizes more the CDFs whose quantile functions deviate negatively from the benchmark quantile function, $\underline{V_\alpha}$ increases with $k$ while $\overline{V_\alpha}$ remains unchanged. Under the setting of this example, we can easily verify that the DM would not change her choice of the indemnity function. 
\begin{figure}[!htbp]
\centering
\minipage{0.5\textwidth}
 \centering
  \includegraphics[width=\linewidth]{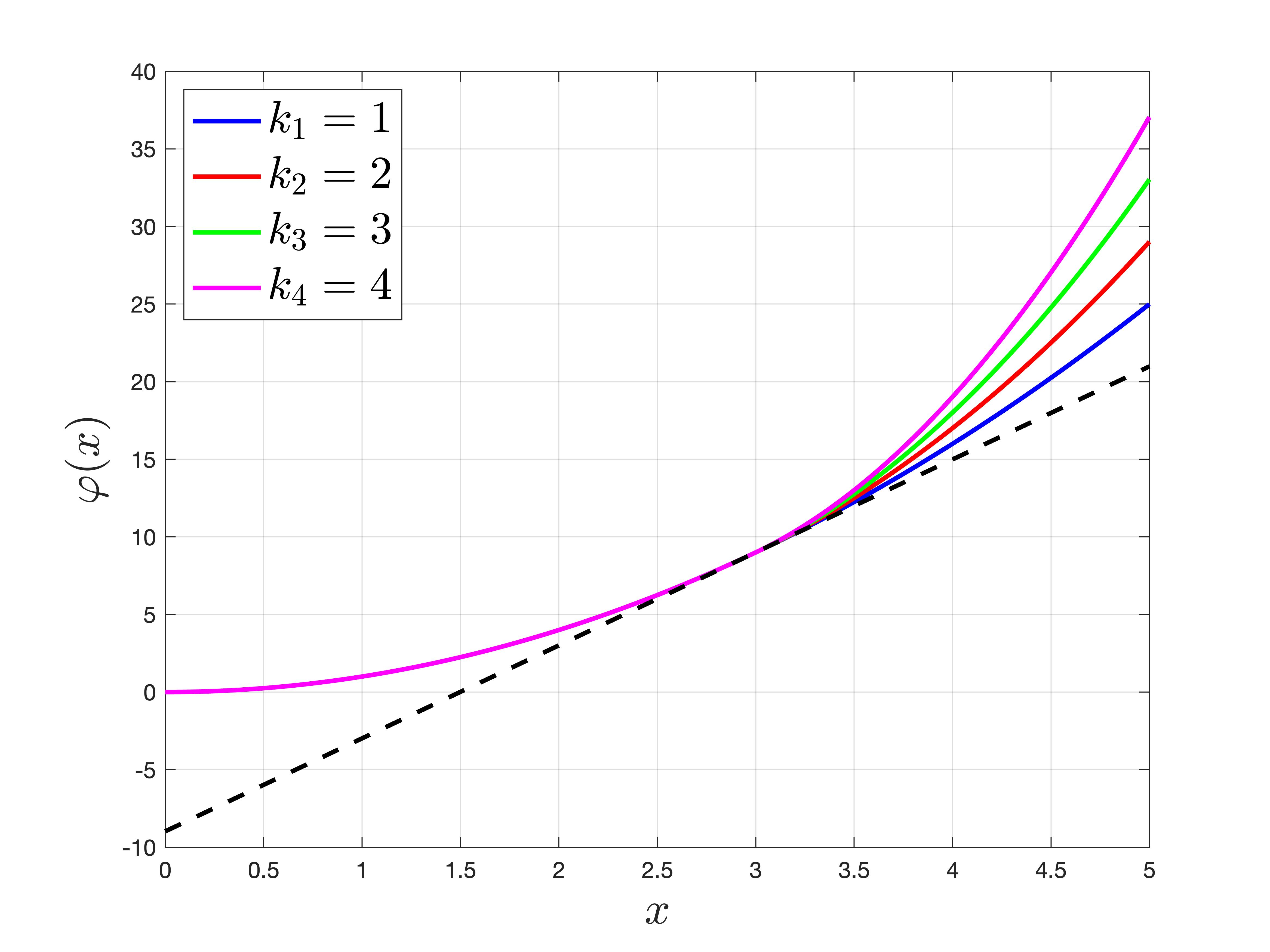}
\endminipage
\minipage{0.5\textwidth}
  \centering
  \includegraphics[width=\linewidth]{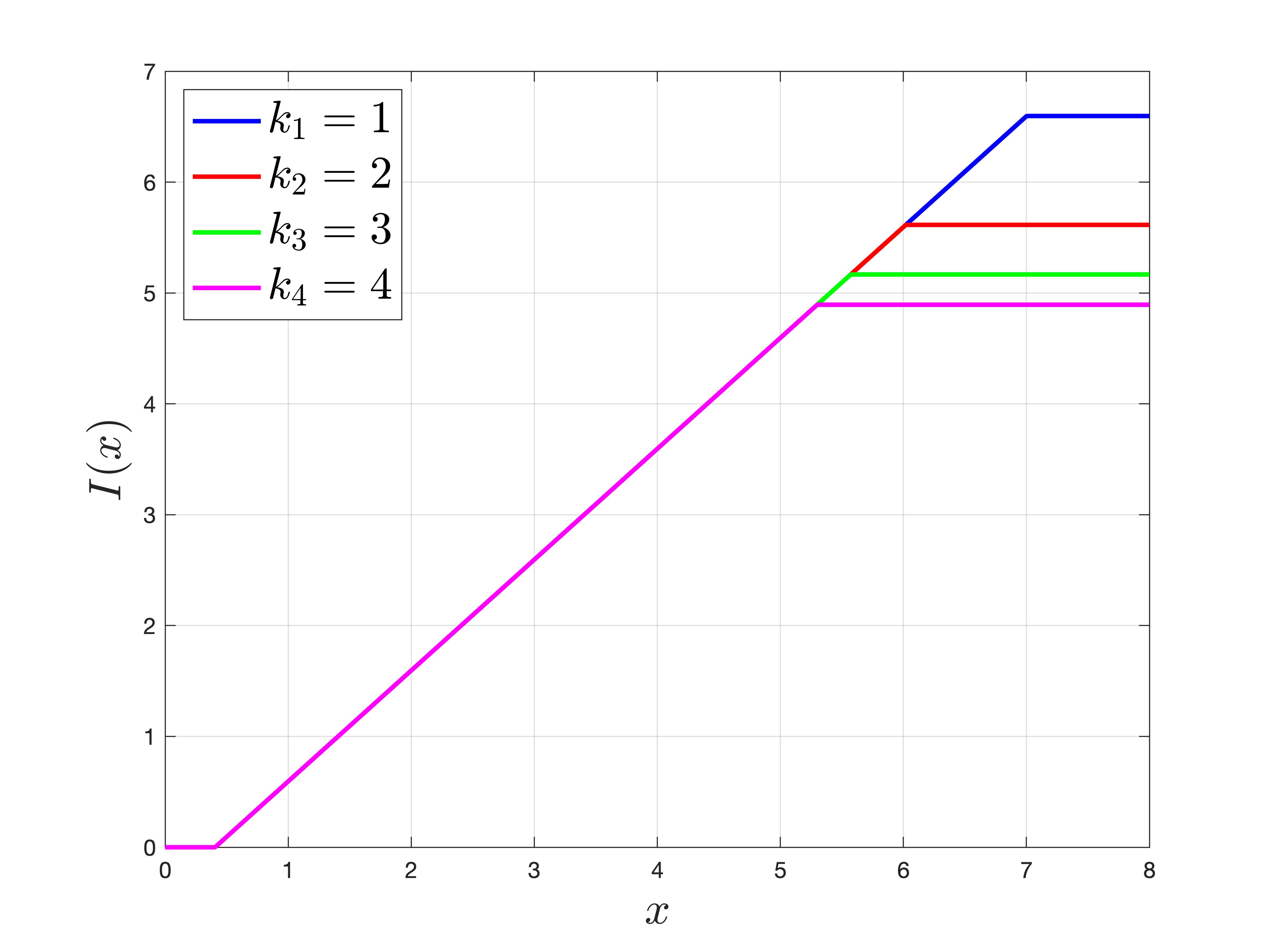}
\endminipage
\caption{(Left) The functions $\varphi$ under different values of $k$; (Right) The corresponding optimal indemnity functions.}
\label{fig:varphi_and_indemnity}
\end{figure}
\end{example}

We remark that in the above example an ambiguity-averse DM is assumed, where $\kappa$ is close to $1$. If the DM is ambiguity-seeking, i.e. a small $\kappa$ is assigned to Problem \ref{Prob:main1}, the resulting conclusion could differ substantially. The closed-form indemnity function in \eqref{eq:explicit_I_var} elegantly encapsulates how these differences in ambiguity preference manifest in the optimal contract.


\section{The solution to Problem \ref{Prob:main2}}\label{sec:sol-prob2}
\subsection{The optimal indemnity function}\label{sec:sol-prob2-indemnity}
With the results developed in Section \ref{sec:sol-prob1}, we can rewrite the constraint of Problem \ref{Prob:main2} as
\begin{equation}\label{Prob2:constraint-reduced}
    \pi(I(X))-I(\overline{V_\alpha})\le A-\overline{V_\alpha}.
\end{equation}
To solve Problem \ref{Prob:main2}, the following assumptions are helpful.
\begin{assumption}\label{Assum:non-trivial}
Throughout the remainder of this section, we impose the following assumptions:
\begin{itemize}
    \item[(\rmnum{1}).] $\inf_{I\in\mathcal{I}}\ \pi(I(X))-I(\overline{V_\alpha})<A-\overline{V_\alpha}$.
    \item[(\rmnum{2}).] {The reference CDF $F_0$ is continuous on $[0,M]$.}
\end{itemize}
\end{assumption}
\noindent If Assumption \ref{Assum:non-trivial} (\rmnum{1}) is violated, the solution to Problem \ref{Prob:main2} either reduces trivially to the minimizer of $\inf_{I\in\mathcal{I}}\ \pi(I(X))-I(\overline{V_\alpha})$ or fails to exist. Assumption \ref{Assum:non-trivial} (\rmnum{2}) is needed to establish the attainability of the solution constructed in our main result.

With Assumption \ref{Assum:non-trivial}, a solution to Problem \ref{Prob:main2} exists, which is presented in the next theorem.
\begin{lemma}\label{thm:sol-exist-prob2}
    Let Assumption \ref{Assum:non-trivial} hold. There exists a solution to Problem \ref{Prob:main2}.
\end{lemma}

Note that $I\mapsto \sup_{F\in\mathcal{B}(F_0,\epsilon)}\rho_g^F(X-I(X)+\pi(I(X)))$ is convex and the constraint \eqref{Prob2:constraint-reduced} is linear in $I$. Under Assumption \ref{Assum:non-trivial}, the Slater's condition holds, which makes solving Problem \ref{Prob:main2} equivalent to solving 
\begin{equation}\label{Prob:main2-v2}
    \inf_{I\in\mathcal{I}}\left\{\sup_{F\in\mathcal{B}(F_0,\epsilon)}\rho^F_g(X-I(X)+\pi(I(X)))+\lambda(\pi(I(X))-I(\overline{V_\alpha})-A+\overline{V_\alpha})\right\}
\end{equation}
for some $\lambda\ge 0$.

Before proceeding, we recall the following well-known minimax theorem.

\begin{lemma}[Minimax theorem \citep{fan1953minimax}]\label{thm:minimax}
    Let $\Xi_1$ be a non-empty convex compact subset of a Hausdorff topological vector space and $\Xi_2$ be a convex set. If $\mathcal{H}$ is a real-valued function defined on $\Xi_1\times\Xi_2$ such that
    \begin{itemize}
        \item[(a).] $\xi_1\mapsto \mathcal{H}(\xi_1,\xi_2)$ is convex and lower semi-continuous on $\Xi_1$ for each $\xi_2\in\Xi_2$;
        \item[(b).] $\xi_2\mapsto\mathcal{H}(\xi_1,\xi_2)$ is concave on $\Xi_2$ for each $\xi_1\in\Xi_1$,
    \end{itemize}
    then $\inf_{\xi_1\in\Xi_1}\sup_{\xi_2\in\Xi_2}\mathcal{H}(\xi_1,\xi_2)=\sup_{\xi_2\in\Xi_2}\inf_{\xi_1\in\Xi_1}\mathcal{H}(\xi_1,\xi_2)$.
\end{lemma}
Note that the set $\mathcal{I}$ is compact under the $L^\infty$ metric (see, for example, \cite{ZHANG2026103202}). Moreover, the convexity of $\mathcal{I}$ and $\mathcal{B}(F_0,\epsilon)$ and the conditions (a) and (b) can be easily verified. Therefore, Lemma \ref{thm:minimax} is applicable here, which turns Problem~\eqref{Prob:main2-v2} into
\begin{equation}\label{prob:minmax}
    \sup_{F\in\mathcal{B}(F_0,\epsilon)}\inf_{I\in\mathcal{I}}\ \left\{\rho_g^F(X-I(X)+\pi(I(X)))+\lambda\left(\pi(I(X))-I(\overline{V_\alpha})-A+\overline{V_\alpha}\right)\right\}
\end{equation}
for which the objective function can be simplified as
\begin{align*}
     & \rho_g^F(X-I(X)+\pi(I(X)))+\lambda\left(\pi(I(X))-I(\overline{V_\alpha})-A+\overline{V_\alpha}\right)\\
    =\ &\rho_g^F(X)-\lambda(A-\overline{V_\alpha})+(1+\lambda)\pi(I(X))-\rho_g^F(I(X))-\lambda I(\overline{V_\alpha}) \\
    =\ &\rho_g^F(X)-\lambda(A-\overline{V_\alpha})+\int_0^M\left\{(1+\lambda)(1+\theta)S_\Q(x)-g(S(x))-\lambda\mathds{1}_{[0,\overline{V_\alpha}]}(x)\right\}dI(x).
\end{align*}
For given $F\in\mathcal{B}(F_0,\epsilon)$ and $\lambda\ge 0$, the optimal indemnity function can be derived straightforwardly. For that purpose, we define 
 \begin{equation*}
 \mathcal{S}(S_0,\epsilon)=\{S:S(x)=1-F(x),\ \forall x\in[0,M]\ \text{with}\ F\in\mathcal{B}(F_0,\epsilon)\},
  \end{equation*}
 where $S_0(x)=1-F_0(x)$, and
 \begin{equation*}
H(x;S,\lambda)=(1+\lambda)(1+\theta)S_\Q(x)-g(S(x))-\lambda\mathds{1}_{[0,\overline{V_\alpha}]}(x).
 \end{equation*} 
We summarize the result of the optimal indemnity function in the following theorem. The proof follows the standard marginal indemnification function (MIF) method employed in \cite{zhuang2016marginal} and is therefore omitted here.
\begin{proposition}\label{thm:op-indemnity}
    Let $S\in\mathcal{S}(S_0,\epsilon)$ and $\lambda\ge 0$ be given. The optimal indemnity function that solves the inner problem of \eqref{prob:minmax}, that is,
    \begin{equation}
        \inf_{I\in\mathcal{I}}\ \left\{\rho_g^F(X-I(X)+\pi(I(X)))+\lambda\left(\pi(I(X))-I(\overline{V_\alpha})-A+\overline{V_\alpha}\right)\right\}
    \end{equation}
    is given by $I^*(x;S,\lambda)=\int_0^x{I^*}'(t;S,\lambda)\,dt$, where
    \begin{equation}
        {I^*}'(t;S,\lambda)=\mathds{1}_{\{t: H(t;S,\lambda)<0\}}(t)+\tilde{\eta}(t)\mathds{1}_{\{t: H(t;S,\lambda)=0\}}(t),
    \end{equation}
    where $\tilde{\eta}(t)$ is a $[0,1]$-valued Lebesgue measurable function.
\end{proposition}
\noindent We interpret $H(x;S,\lambda)$ as the net price of marginal indemnity, and the DM will always purchase the full marginal indemnity when the net price is negative and is indifferent between with and without insurance when this price is zero.

With Proposition \ref{thm:op-indemnity}, the objective function of Problem~\eqref{Prob:main2-v2} can be simplified as
\begin{align*}
    &\rho_g^F(X)-\lambda(A-\overline{V_\alpha})+\int_0^M\left\{(1+\lambda)(1+\theta)S_\Q(x)-g(S(x))-\lambda\mathds{1}_{[0,\overline{V_\alpha}]}(x)\right\}dI^*(x;S,\lambda) \\
    =\ &\rho_g^F(X)-\lambda(A-\overline{V_\alpha})-\int_0^M\left(g(S(x))+\lambda\mathds{1}_{[0,\overline{V_\alpha}]}(x)-(1+\lambda)(1+\theta)S_\Q(x)\right)_+\,dx \\
    =\ &\int_0^Mg(S(x))\wedge\left((1+\lambda)(1+\theta)S_\Q(x)-\lambda\mathds{1}_{[0,\overline{V_\alpha}]}(x)\right)\,dx-\lambda(A-\overline{V_\alpha}).
\end{align*}
This reduces Problem~\eqref{Prob:main2-v2} to
\begin{align}
    \sup_{S\in\mathcal{S}(S_0,\epsilon)}\int_0^Mg(S(x))\wedge\left((1+\lambda)(1+\theta)S_\Q(x)-\lambda\mathds{1}_{[0,\overline{V_\alpha}]}(x)\right)\,dx-\lambda(A-\overline{V_\alpha}). \label{Prob:main2-v4}
\end{align}
In the next section, we characterize the worst-case survival function that solves problem \eqref{Prob:main2-v4}.

\subsection{The worst-case distribution}\label{sec:worst-dist}
It is now sufficient to focus on the problem
\begin{equation}\label{Prob:inner}
    \sup_{S\in\mathcal{S}(S_0,\epsilon)}\ \int_0^M g(S(x))\wedge\left((1+\lambda)(1+\theta)S_\Q(x)-\lambda\mathds{1}_{[0,\overline{V_\alpha}]}(x)\right)\,dx
\end{equation}
for some $\lambda\ge 0$. Note that the objective function of \eqref{Prob:inner} is monotone in $S(x)$, which directly yields the following lemma. 
\begin{lemma}
The worst-case survival function $S^*(x)$ that solves Problem~\eqref{Prob:inner}, if exists, must satisfy $S^*(x)\ge S_0(x)$ for all $x\in[0,M]$,
whose corresponding quantile function satisfies $F^{*-1}(t)\ge F_0^{-1}(t)$ for all $t\in[0,1]$.
\end{lemma}
\noindent Therefore, for Problem~\eqref{Prob:inner} we can restrict ourselves to the following smaller uncertainty set:
\begin{equation*}
\tilde{\mathcal S}(S_0,\epsilon) := 
\big\{S\in\mathcal S(S_0,\epsilon): S(x)\ge S_0(x),\ \forall x\in[0,M]\big\}.
\end{equation*}

With the help of Lemma $\ref{lem:survival-rep-BW}$, the BW constraint can be rewritten as
\begin{equation}\label{eq:BW_constraint}
\begin{aligned}
   &\int_0^M \Big(\big[\varphi'(x)-\varphi'(0)\big]S(x)-\int_0^{M}\varphi''(y)\big(S_0(y)\wedge S(x)\big)\,dy\Big)\,dx \le \zeta,
\end{aligned}
\end{equation}
where $\zeta := \epsilon-\int_0^M\varphi'(y)yS_0(y)\,dy$. Since the objective in $\eqref{Prob:inner}$ is concave in $S(\cdot)$, and the constraint is convex in $S(\cdot)$, with $S_0$ being a strictly feasible solution, strong duality holds because Slater's conditions are fulfilled. Similar to Lemma \ref{thm:sol-exist-prob2}, we first verify the existence of a solution to Problem \eqref{Prob:inner}.
\begin{lemma}\label{thm:sol-exist-inner}
    Given any $\lambda\ge 0$, there exists a solution to Problem \eqref{Prob:inner}.
\end{lemma}
\noindent To this end, solving Problem~\eqref{Prob:inner} boils down to solving 
\begin{equation}\label{eq:totalLg}
\begin{aligned}
\sup_{S\in\tilde{\mathcal S}(S_0,\infty)}\ \mathscr{L}(S(x);\beta):= 
\int_0^{M} \bigg\{&g(S(x))\wedge\left((1+\lambda)(1+\theta)S_\Q(x)-\lambda\mathds{1}_{[0,\overline{V_\alpha}]}(x)\right)\\
&-\beta\Big(\big[\varphi'(x)-\varphi'(0)\big]S(x)-\int_0^{M}\varphi''(y)S_0(y)\wedge S(x)\,dy\Big)
\bigg\}\,dx
\end{aligned}
\end{equation}
for some $\beta\ge 0$.

While the objective function in \eqref{eq:totalLg} is concave in $S(\cdot)$, it is discontinuous in $x$. Consequently, the pointwise maximizer may contain an upward jump at the point of discontinuity, violating the definition of a survival function. We thus temporarily deviate from Problem~\eqref{eq:totalLg} and first solve the following auxiliary problem:
\begin{equation}\label{eq:totalLg-e}
\sup_{G\in\mathcal{G}}\ \mathscr{L}(G(x);\beta)
\end{equation}
where 
 \begin{equation*}
 \mathcal{G}:=\{G: G(x)\in[S_0(x),1], \ \forall x\in[0,M]\}.
  \end{equation*}
We denote by $G^*(x;\beta)$ the solution to Problem \eqref{eq:totalLg-e}. We will show later that the solution to Problem~\eqref{eq:totalLg} can be obtained via a relaxation approach based on the solution to Problem~\eqref{eq:totalLg-e}.

To solve Problem~\eqref{eq:totalLg-e}, we define
\begin{equation*} 
x_1 := \sup\bigg\{x\in[0,\overline{V_\alpha}): S_{\mathbb Q}(x)\ge \tfrac{1}{1+\theta}\bigg\},\quad x_2:= \sup\bigg\{x\in[\overline{V_\alpha},M): S_{\mathbb Q}(x)\ge \tfrac{1}{(1+\lambda)(1+\theta)}\bigg\}
\end{equation*}
and
\begin{equation*} 
\begin{aligned} 
&\mathscr{A}_1 := \{x\in[x_1,\overline{V_\alpha}):H(x;S_0,\lambda)\ge 0\},\quad \mathscr{B}_1 := [x_1,\overline{V_\alpha})\setminus\mathscr{A}_1,\\ 
&\mathscr{A}_2 := \{x\in[x_2,M]:H(x;S_0,\lambda)\ge 0\},\quad \mathscr{B}_2 := [x_2,M]\setminus\mathscr{A}_2. 
\end{aligned} 
\end{equation*}
The following proposition characterizes the worst-case functions in $\mathcal{G}$ that solve Problem \eqref{eq:totalLg-e}.

\begin{proposition}\label{prop:beta>0}
The solution to Problem $\eqref{eq:totalLg-e}$ is given by
\begin{equation*}
\begin{aligned}
  G^*(x;\beta)=&\hat G(x;\beta)\mathds{1}_{[0,x_1)}(x)
+\big(\hat G(x;\beta)\wedge g^{-1}\left((1+\lambda)(1+\theta)S_\Q(x)-\lambda\right)\big)\mathds{1}_{\mathscr{A}_1}(x)
+S_0(x)\mathds{1}_{\mathscr{B}_1}(x)\\
&+\hat G(x;\beta)\mathds{1}_{[\overline{V_\alpha},x_2)}(x)
+\big(\hat G(x;\beta)\wedge g^{-1}\left((1+\lambda)(1+\theta)S_\Q(x)\right)\big)\mathds{1}_{\mathscr{A}_2}(x)
+S_0(x)\mathds{1}_{\mathscr{B}_2}(x),  
\end{aligned}
\end{equation*}
where
\begin{equation*}
\hat G(x;\beta)=\inf\{t\in[S_0(x),1]: K'(t;x)\le 0\}
\end{equation*}
and
\begin{equation*}
K(t;x):=g(t)-\beta\!\left(\big[\varphi'(x)-\varphi'(0)\big]t-\int_0^{M}\varphi''(y)S_0(y)\wedge  t\,dy\right).
\end{equation*}
Furthermore, $\hat{G}(x;0)=g^{-1}(1)\vee S_0(x)$ and
\begin{equation*}
    G^*(\cdot;0)\in \arg\inf_{G\in\mathcal{G}_0}\ \int_0^M \Big(\big[\varphi'(x)-\varphi'(0)\big]G(x)-\int_0^{M}\varphi''(y)\big(S_0(y)\wedge G(x)\big)\,dy\Big)\,dx,
\end{equation*}
where 
 \begin{equation*}
 \mathcal{G}_0=\{\tilde{G}\in\mathcal{G}:\ \mathscr{L}(\tilde{G}(x);0)=\sup_{G\in\mathcal{G}}\ \mathscr{L}(G(x);0)\}.
  \end{equation*}
\end{proposition}

One can verify that $G^*(x;\beta)$ is decreasing on the intervals $[0,\overline{V_\alpha})$ and $[\overline{V_\alpha},M]$, respectively. However, it may fail to be decreasing on the entire interval $[0,M]$. For instance, if $\overline{V_\alpha}-\in \mathscr{A}_1$, $\overline{V_\alpha}\in \mathscr{A}_2$, and $S_\Q$ is continuous at $\overline{V_\alpha}$, then
\begin{align*}
 G^*(\overline{V_\alpha}-;0)
= g^{-1}\big((1+\lambda)(1+\theta)S_{\mathbb Q}(\overline{V_\alpha})-\lambda\big)
\le G^*(\overline{V_\alpha};0)
= g^{-1}\big((1+\lambda)(1+\theta)S_{\mathbb Q}(\overline{V_\alpha})\big),
\end{align*}
and for $\beta>0$,
\begin{align*}
&G^*(\overline{V_\alpha}-;\beta)
= \hat G(\overline{V_\alpha};\beta)\wedge g^{-1}\big((1+\lambda)(1+\theta)S_{\mathbb Q}(\overline{V_\alpha})-\lambda\big)\\
\le&\ G^*(\overline{V_\alpha};\beta)
=\hat G(\overline{V_\alpha};\beta)\wedge g^{-1}\big((1+\lambda)(1+\theta)S_{\mathbb Q}(\overline{V_\alpha})\big).
\end{align*}
If such an upward jump exists, i.e., $G^*(\overline{V_\alpha}-;\beta)<G^*(\overline{V_\alpha};\beta)$, the worst-case survival function solving Problem \eqref{eq:totalLg} can be derived via the following constructive relaxation procedure.\footnote{The relaxation method first replaces the full admissibility conditions with pointwise constraints and solves the resulting relaxed problem. Pointwise optimization often yields an explicit or analytically
tractable candidate, which is then checked against the defining
properties of a survival function, particularly monotonicity and
right-continuity. If necessary, suitable modifications are applied to
restore admissibility.} Let
\begin{align}\label{eq:modified-S}
S(x;b(\beta))=&\ \left(G^*(x;\beta)\vee b(\beta)\right)\mathds{1}_{[0,\overline{V_\alpha}]}(x)+\left(G^*(x;\beta)\wedge b(\beta)\right)\mathds{1}_{[\overline{V_\alpha},M]}(x) \\
=&\ \begin{cases}
G^*(x;\beta), & x\in [0,a_1(\beta))\cup [a_2(\beta),M],\\
b(\beta), & x\in [a_1(\beta),a_2(\beta)),
\end{cases} \nonumber
\end{align}
where
\begin{equation*}
a_1(\beta):=\sup\big\{x\in[0,\overline{V_\alpha}): G^*(x;\beta)\ge b(\beta)\big\},\quad\
a_2(\beta):=\sup\big\{x\in[\overline{V_\alpha},M): G^*(x;\beta)\ge b(\beta)\big\}.
\end{equation*}

Note that $S(x;b(\beta))$ is a well-defined survival function on $[0,M]$. The following theorem establishes the optimality of \eqref{eq:modified-S} for Problem \eqref{eq:totalLg} and presents the solution to Problem \eqref{Prob:inner}.
\begin{theorem}\label{th:beta}
Let
\[
O_b(\beta)=
\begin{cases}
[G^*(\overline{V_\alpha}-;\beta),G^*(\overline{V_\alpha};\beta)],
& \text{if } G^*(\overline{V_\alpha}-;\beta)<G^*(\overline{V_\alpha};\beta),\\
\{G^*(\overline{V_\alpha};\beta)\},
& \text{if } G^*(\overline{V_\alpha}-;\beta)\geq G^*(\overline{V_\alpha};\beta).
\end{cases}
\]
For any survival function
$S\in\tilde{\mathcal{S}}(S_0,\infty)$,
there exists a $b(\beta)\in O_b(\beta)$ such that
$S(\cdot;b(\beta))\in\tilde{\mathcal{S}}(S_0,\infty)$ and
\[
\mathcal{L}(S(x;b(\beta));\beta)
\geq
\mathcal{L}(S(x);\beta),
\]
{where, for $\beta>0$, the inequality is strict when the Lebesgue measure of $\{x\in[0,M]: S(x;b(\beta))\neq S(x)\}$ is positive.} Let
\begin{equation}\label{eq:B-star}
B^*(\beta):=\arg\sup_{b(\beta)\in O_b(\beta)}\mathcal{L}(S(x;b(\beta));\beta).
\end{equation}
Then, for any $b^*(\beta)\in B^*(\beta)$,
$S(x;b^*(\beta))$ solves Problem~\eqref{eq:totalLg}.
{Furthermore, if $S(x;G^*(\overline{V_\alpha};0))$ satisfies \eqref{eq:BW_constraint}, then $\beta^*=0$, and $S(x;G^*(\overline{V_\alpha};0))$ solves Problem~\eqref{Prob:inner}.
Otherwise, there exist $\beta^*>0$ and
$b^*(\beta^*)\in B^*(\beta^*)$ such that
\eqref{eq:BW_constraint} holds with equality for
$S(x;b^*(\beta^*))$, and $S(x;b^*(\beta^*))$ solves
Problem~\eqref{Prob:inner}.}
\end{theorem}

\begin{figure}[!htbp]
\centering
\minipage{0.3\textwidth}
 \centering
  \includegraphics[width=\linewidth]{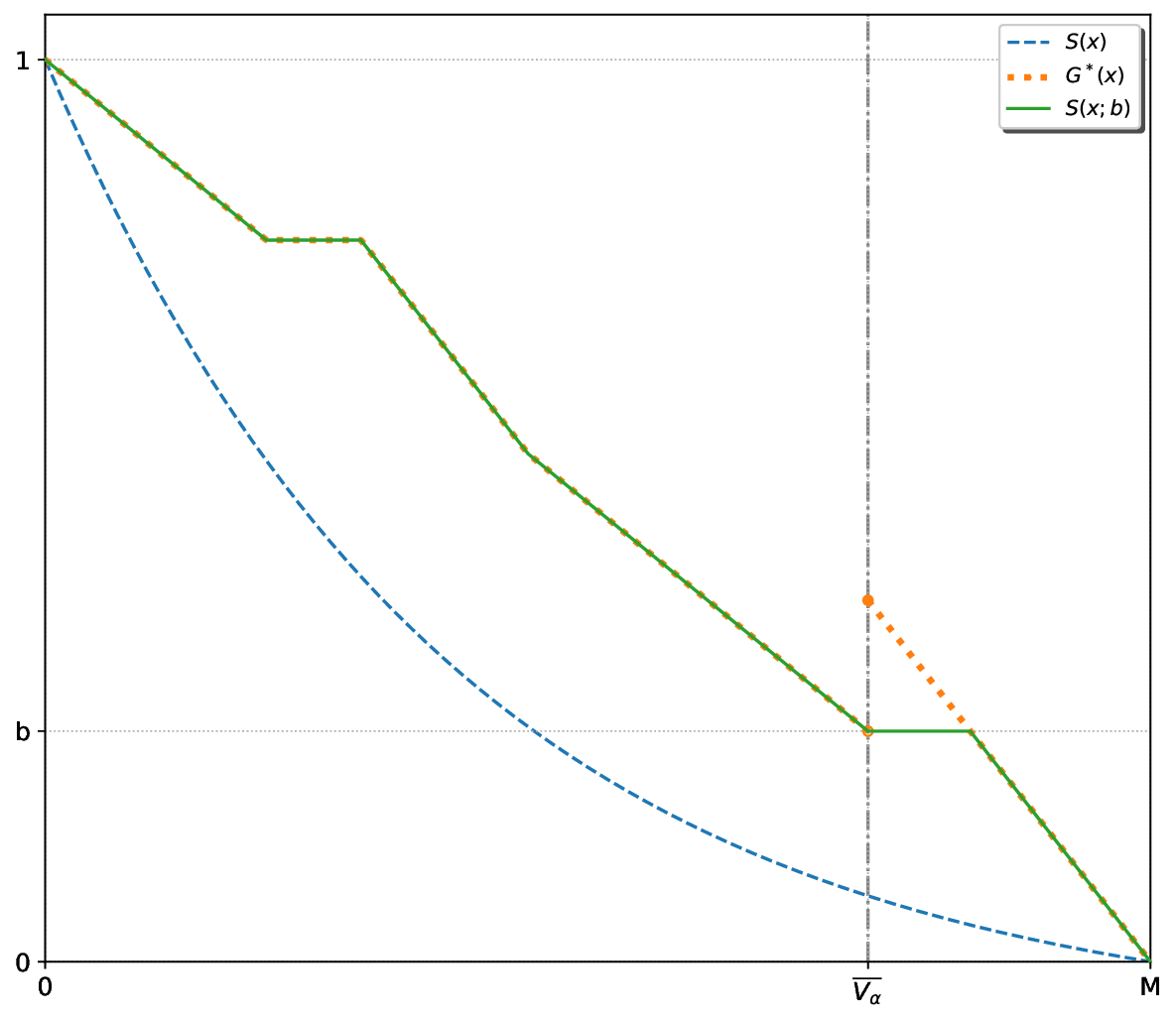}
\endminipage
\hspace{0.2cm}
\minipage{0.3\textwidth}
  \centering
  \includegraphics[width=\linewidth]{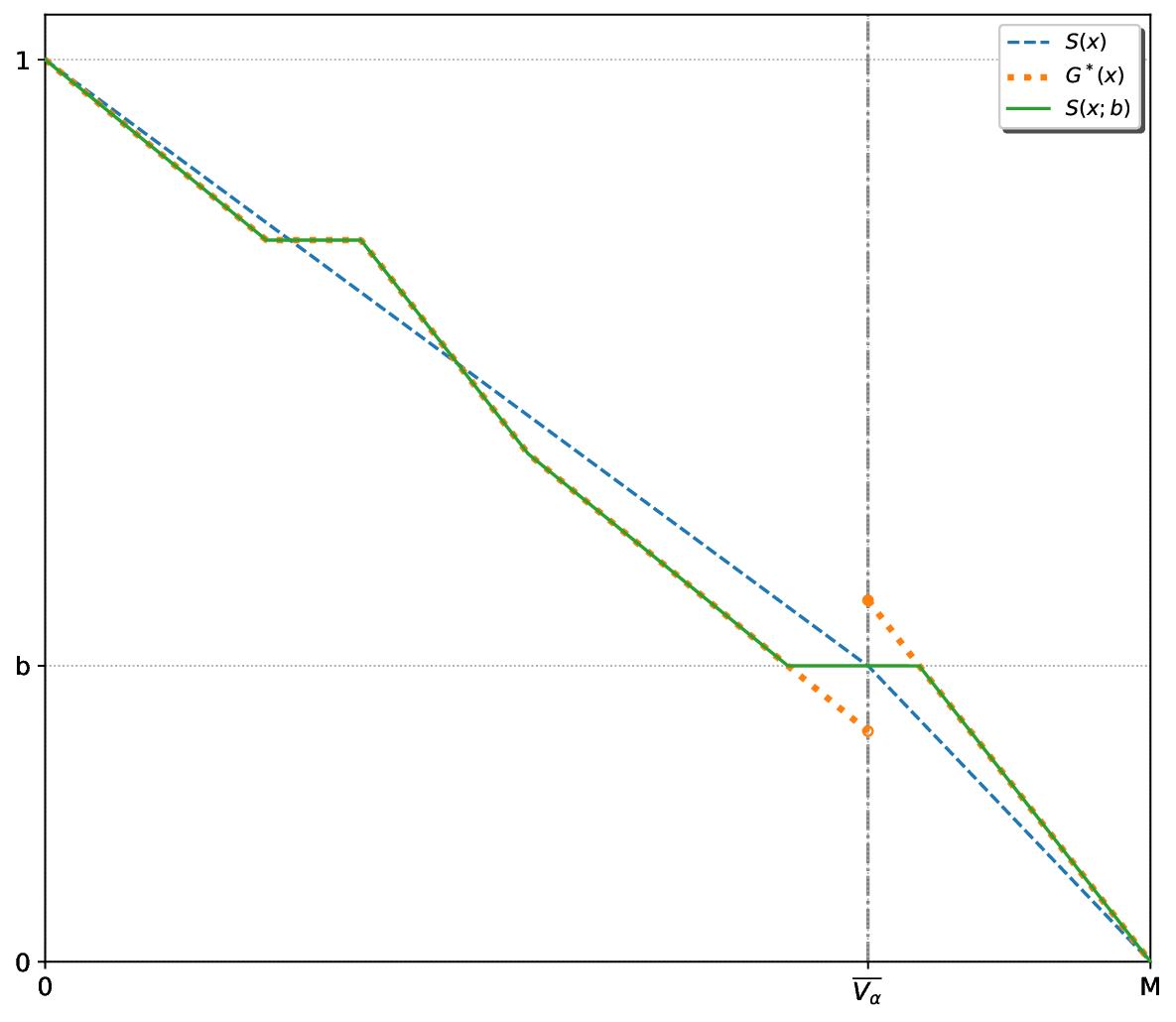}
\endminipage
\hspace{0.2cm}
\minipage{0.3\textwidth}
  \centering
  \includegraphics[width=\linewidth]{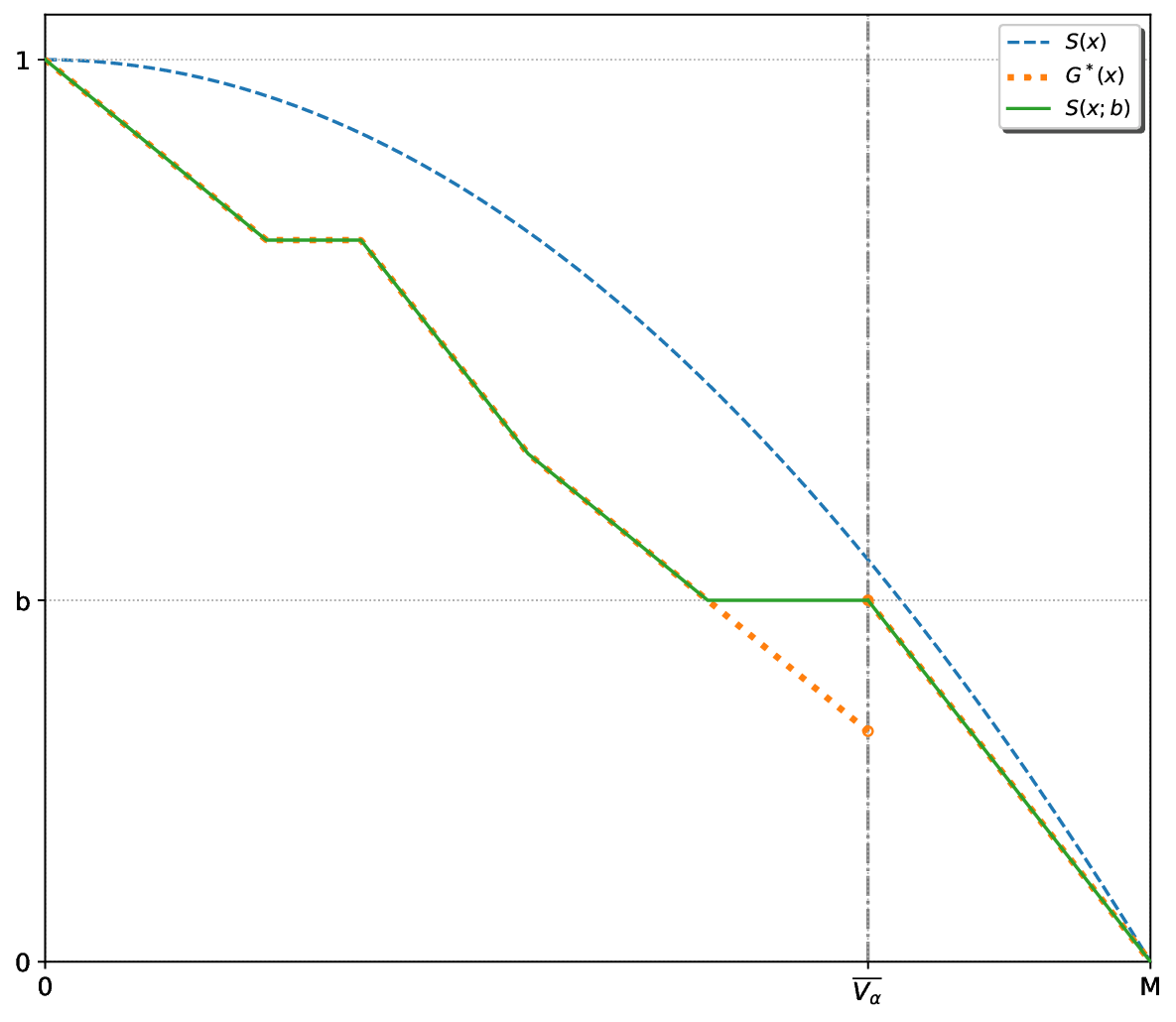}
\endminipage
\caption{Graphical illustration of Theorem~\ref{th:beta} when $\beta^*=0$: (Left) $S(\overline{V_\alpha})\in\big(G^*(\overline{V_\alpha}-;0),\,G^*(\overline{V_\alpha};0)\big)$; (Middle) $S(\overline{V_\alpha})\le G^*(\overline{V_\alpha}-;0)$; (Right) $S(\overline{V_\alpha})\ge G^*(\overline{V_\alpha};0)$.}
\label{fig:S_xb}
\end{figure}

Figure~\ref{fig:S_xb} provides a graphical illustration of Theorem~\ref{th:beta} when $\beta^*=0$. We are now prepared to present the interpretations of the results of Theorem \ref{th:beta}, with a focus on the case where $\beta^*=0$, along with findings regarding the effect of $\lambda$ on the worst-case distribution. The results for the case where $\beta^*>0$ can be interpreted similarly. Since the worst-case survival function is derived from the solution to \eqref{eq:totalLg-e}, key insights focus on the behavior and characteristics of $G^*(x; \beta)$.
\begin{itemize}
    \item[(a).] By \eqref{Prob:inner}, the DM has no incentive to exaggerate the loss if $g(S_0(x))=1$, resulting in the worst-case survival function $g^{-1}(1)\vee S_0(x)$ over $[0,x_1)\cup[\overline{V_\alpha},x_2)$.
    \item[(b).] As explained in Section \ref{sec:sol-prob2-indemnity}, the DM purchases full marginal insurance when the net premium $H(x;S,\lambda)$ is negative. Since $H(x;S,\lambda)\le H(x;S_0,\lambda)$ for $S\in\tilde{\mathcal S}(S_0,\epsilon)$, the DM has no incentive to consider the worst-case losses in $\mathscr{B}_1\cup\mathscr{B}_2$ as those losses are ceded to the insurer.
    \item[(c).] The net premium $H(x;S,\lambda)$ increases with $\lambda$ when $x\in[0,x_1)\cup[\overline{V_\alpha},M]$ but decreases with $\lambda$ when $x\in[x_1,\overline{V_\alpha})$. Consequently, the incentive to insure the loss is reduced over the former intervals and increased over the latter when $\lambda$ increases. This results in smaller $\mathscr{A}_1$ and $\mathscr{B}_2$ but larger $\mathscr{B}_1$ and $\mathscr{A}_2$. Thus, the DM tends to exaggerate large losses (beyond the worst-case VaR) while staying with the benchmark distribution for small ones under a tighter constraint. 
    \item[(d).] The finding of (c) explains the possible upward jump of $G^*(x;\beta)$ at $\overline{V_\alpha}$. Note that the upward jump occurs exclusively when $\lambda>0$, i.e., when the constraint becomes binding. A binding VaR constraint, which restricts the distribution only by the value of its $\alpha$-quantile and not by the allocation of probability mass around it, directly produces the flat segment observed in $S(x;b)$.
\end{itemize}

Note that Problem \eqref{Prob:inner} depends on $\lambda$. We now abuse the notation a bit and write the solution to Problem \eqref{Prob:inner} as $S_\lambda(x;b_\lambda^*(\beta_\lambda^*))$. We end this section by presenting the following Theorem that contains the solution to Problem \ref{Prob:main2}.

\begin{theorem}\label{thm:main-prob2}
    The optimal insurance indemnity function that solves Problem \ref{Prob:main2} is given by $I^*(x;S_{\lambda^*}(x;b_{\lambda^*}^*(\beta^*_{\lambda^*})),\lambda^*)$, as established in Proposition \ref{thm:op-indemnity}, and the corresponding worst-case survival function is given by $S_{\lambda^*}(x;b_{\lambda^*}^*(\beta^*_{\lambda^*}))$, as demonstrated in Theorem \ref{th:beta}. The function $\tilde{\eta}$ appearing in Proposition \ref{thm:op-indemnity} and the Lagrangian parameter $\lambda^*$ are chosen to satisfy the Karush-Kuhn-Tucker (KKT) conditions
    \begin{align*}
        \lambda^*(\pi(I^*(X;S_{\lambda^*}(x;b_{\lambda^*}^*(\beta^*_{\lambda^*})),\lambda^*))-I^*(\overline{V_\alpha};S_{\lambda^*}(x;b_{\lambda^*}^*(\beta^*_{\lambda^*})),\lambda^*)-A+\overline{V_\alpha})=0
    \end{align*}
    and
    \begin{align*}
        \pi(I^*(X;S_{\lambda^*}(x;b_{\lambda^*}^*(\beta^*_{\lambda^*})),\lambda^*))-I^*(\overline{V_\alpha};S_{\lambda^*}(x;b_{\lambda^*}^*(\beta^*_{\lambda^*})),\lambda^*)\le A-\overline{V_\alpha}.
    \end{align*}
\end{theorem}
In view of Lemma \ref{thm:sol-exist-prob2}, together with the strong duality for Problem \ref{Prob:main2}, the existence of the Lagrange multiplier $\lambda^*$, as well as the associated KKT conditions, follows from Theorem 1 in Section 8.3 of \cite{luenberger1997optimization}.

\subsection{A concrete example under TVaR}\label{sec:concrete}
This section presents a specific example of Problem \eqref{Prob:inner} with $\rho_g(\cdot)=\mathrm{TVaR}_\alpha(\cdot)$. To ease our discussions, we assume that $S_0=S_\Q$, i.e., the DM applies the loss distribution under the insurer's subjective belief as the benchmark distribution, and that $S_\Q$ has continuous support over $[0,M]$. In practice, the confidence level $\alpha$ for $\mathrm{TVaR}_\alpha(\cdot)$ is set to be large (close to $1$). Hence, we focus on the case where $\frac{1}{1+\theta}>1-\alpha$ in this example. To derive the worst-case distribution, we first need to compute $G^*(x;\beta)$, which involves calculating $\hat{G}(x;\beta)$ and partitioning the domain $[0,M]$ into the sub-domains $[0,x_1)$, $\mathscr{A}_1$, $\mathscr{B}_1$, $[\overline{V_\alpha},x_2)$, $\mathscr{A}_2$ and $\mathscr{B}_2$. 

After simple calculations, the right-hand derivative of the function $K(t;x)$ is 
 \begin{equation*}
K'(t;x)=\left\{
\begin{aligned}
&-\beta[\varphi'(x)-\varphi'(F_\Q^{-1}(1-t))],&\quad& t\ge 1-\alpha, \\
&\frac{1}{1-\alpha}-\beta[\varphi'(x)-\varphi'(F_\Q^{-1}(1-t))],&\quad& t< 1-\alpha, 
\end{aligned}
\right.
 \end{equation*}
which yields 
 \begin{equation*}
\hat{G}(x;\beta)=S_\Q(x)\mathds{1}_{[0,\tilde{x}_\alpha^\Q)}(x)+(1-\alpha)\mathds{1}_{[\tilde{x}_\alpha^\Q,\hat{x})}(x)+\hat{S}_\Q(x)\mathds{1}_{[\hat{x},M]}(x),
 \end{equation*}
where $\tilde{x}_\alpha^\Q=\mathrm{VaR}_\alpha^{F_\Q}(X)$, $\hat{x}=[\varphi']^{-1}(\varphi'(\tilde{x}_\alpha^\Q)+\frac{1}{\beta(1-\alpha)})$ and $\hat{S}_\Q(x)=S_\Q([\varphi']^{-1}(\varphi'(x)-\frac{1}{\beta(1-\alpha)}))$.

In our setting, the net price with $S_0=S_\Q$ is $H(x;S_\Q,\lambda)=(1+\lambda)(1+\theta)S_\Q(x)-g(S_\Q(x))-\lambda\mathds{1}_{[0,\overline{V_\alpha}]}(x)$. Figure \ref{fig:TVaR_case123} summarizes the three possible cases, which are elaborated as follows:
\begin{itemize}
    \item \underline{Case (i): $(1+\lambda)(1+\theta)<\frac{1}{1-\alpha}$}.

    This case is illustrated in the first sub-figure of Figure \ref{fig:TVaR_case123}. In this case, we have $H(x;S_\Q,\lambda)<0$ for $x\in(x_1,M)$, which results in 
     \begin{equation*}
    G^*(x;\beta)=\hat{G}(x;\beta)\mathds{1}_{[0,x_1)}(x)+S_\Q(x)\mathds{1}_{[x_1,M]}(x).
     \end{equation*}
    By the definition of $x_1$, we have $S_\Q(x_1)=\frac{1}{1+\theta}>1-\alpha$. Therefore, $x_1<\tilde{x}_\alpha^\Q$. This leads to $G^*(x;\beta)=S_\Q(x)$ for $x\in[0,M]$.
\end{itemize}

\begin{figure}[!htbp]
\centering
\minipage{0.3\textwidth}
 \centering
  \includegraphics[width=\linewidth]{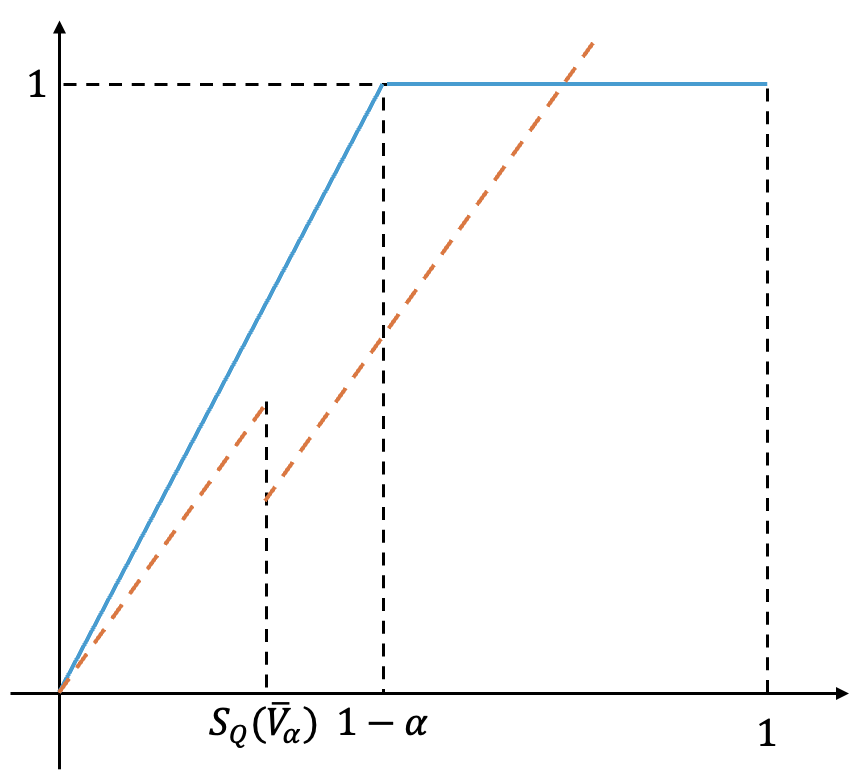}
\endminipage
\hspace{0.2cm}
\minipage{0.3\textwidth}
  \centering
  \includegraphics[width=\linewidth]{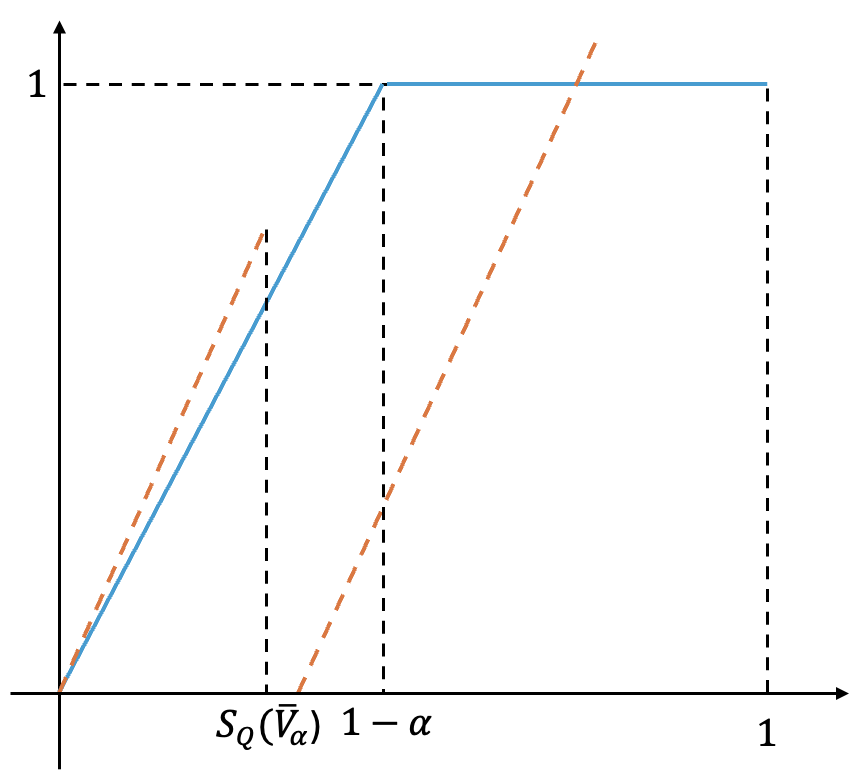}
\endminipage
\hspace{0.2cm}
\minipage{0.3\textwidth}
  \centering
  \includegraphics[width=\linewidth]{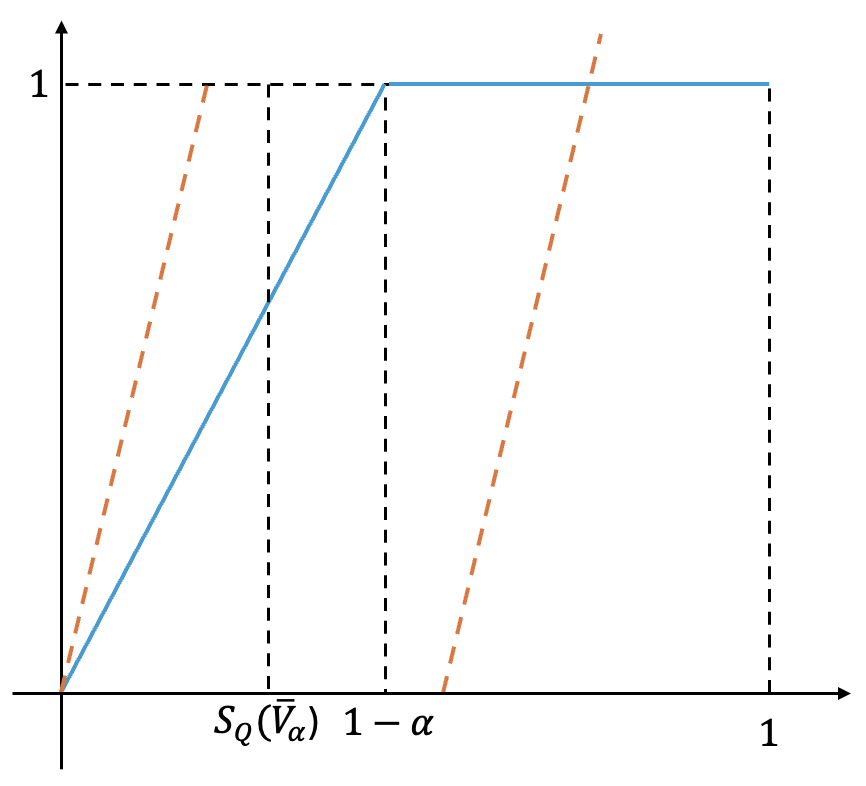}
\endminipage
\caption{The three cases of the function $\hat{H}(t):=(1+\lambda)(1+\theta)t-\lambda\mathds{1}_{[S_\Q(\overline{V_\alpha}),1]}(t)$: (Left) $(1+\lambda)(1+\theta)<\frac{1}{1-\alpha}$; (Middle) $\frac{1}{1-\alpha}\le(1+\lambda)(1+\theta)\le\frac{1}{S_\Q(\overline{V_\alpha})}$; (Right) $(1+\lambda)(1+\theta)>\frac{1}{S_\Q(\overline{V_\alpha})}$.}
\label{fig:TVaR_case123}
\end{figure}

\begin{itemize}
    \item \underline{Case (ii): $\frac{1}{1-\alpha}\le(1+\lambda)(1+\theta)\le\frac{1}{S_\Q(\overline{V_\alpha})}$}.

    This case is illustrated in the second sub-figure of Figure \ref{fig:TVaR_case123}. In this case, the net premium $H(x;S_\Q,\lambda)<0$ over $(x_1,\overline{V_\alpha})$ and $H(x;S_\Q,\lambda)>0$ over $(\overline{V_\alpha},M)$, yielding
     \begin{equation*}
    \begin{aligned}
        G^*(x;\beta)=&\hat{G}(x;\beta)\mathds{1}_{[0,x_1)}(x)+S_\Q(x)\mathds{1}_{[x_1,\overline{V_\alpha})}(x) \\
        &+\hat{G}(x;\beta)\wedge (1-\alpha)(1+\lambda)(1+\theta)S_\Q(x)\mathds{1}_{[\overline{V_\alpha},M]}(x).
    \end{aligned}
     \end{equation*}
    Note that $S_\Q(x_2)\ge S_\Q(\overline{V_\alpha})$ in this specific case, thus, by the definition of $x_2$, we have $(1+\lambda)(1+\theta)S_\Q(\overline{V_\alpha})\le 1$. This results in two possible sub-cases for $G^*(x;\beta)$. If $\hat{x}\le \overline{V_\alpha}$, then
     \begin{equation*}
    G^*(x;\beta)=S_\Q(x)\mathds{1}_{[0,\overline{V_\alpha})}(x)+\hat{S}_\Q(x)\wedge (1-\alpha)(1+\lambda)(1+\theta)S_\Q(x)\mathds{1}_{[\overline{V_\alpha},M]}(x).
     \end{equation*}
    If $\hat{x}>\overline{V_\alpha}$, then 
    \begin{align*}
        G^*(x;\beta)=&S_\Q(x)\mathds{1}_{[0,\overline{V_\alpha})}(x)+(1-\alpha)(1+\lambda)(1+\theta)S_\Q(x)\mathds{1}_{[\overline{V_\alpha},\hat{x})}(x) \\
        &+\hat{S}_\Q(x)\wedge (1-\alpha)(1+\lambda)(1+\theta)S_\Q(x)\mathds{1}_{[\hat{x},M]}(x).
    \end{align*}

    \item \underline{Case (iii): $(1+\lambda)(1+\theta)>\frac{1}{S_\Q(\overline{V_\alpha})}$}.

    This case is illustrated in the third sub-figure of Figure \ref{fig:TVaR_case123}. The only difference between this case and the case (\rmnum{2}) is that $x_2>\overline{V_\alpha}$ here, which yields
    \begin{align*}
        G^*(x;\beta)=&\hat{G}(x;\beta)\mathds{1}_{[0,x_1)}(x)+S_\Q(x)\mathds{1}_{[x_1,\overline{V_\alpha})}(x)+\hat{G}(x;\beta)\mathds{1}_{[\overline{V_\alpha},x_2)}(x) \\
        &+\hat{G}(x;\beta)\wedge (1-\alpha)(1+\lambda)(1+\theta)S_\Q(x)\mathds{1}_{[x_2,M]}(x).
    \end{align*}
    Depending on the location of $\hat{x}$, we have three sub-cases here. If $\hat{x}\le\overline{V_\alpha}$, then 
     \begin{equation*}
    G^*(x;\beta)=S_\Q(x)\mathds{1}_{[0,\overline{V_\alpha})}(x)+\hat{S}_\Q(x)\mathds{1}_{[\overline{V_\alpha},x_2)}(x)+\hat{S}_\Q(x)\wedge (1-\alpha)(1+\lambda)(1+\theta)S_\Q(x)\mathds{1}_{[x_2,M]}(x).
     \end{equation*}
    If $\overline{V_\alpha}<\hat{x}\le x_2$, then
    \begin{align*}
        G^*(x;\beta)=&S_\Q(x)\mathds{1}_{[0,\overline{V_\alpha})}(x)+(1-\alpha)\mathds{1}_{[\overline{V_\alpha},\hat{x})}(x)+\hat{S}_\Q(x)\mathds{1}_{[\hat{x},x_2)}(x) \\
        &+\hat{S}_\Q(x)\wedge (1-\alpha)(1+\lambda)(1+\theta)S_\Q(x)\mathds{1}_{[x_2,M]}(x).
    \end{align*}
    If $x_2<\hat{x}$, then
    \begin{align*}
        G^*(x;\beta)=&S_\Q(x)\mathds{1}_{[0,\overline{V_\alpha})}(x)+(1-\alpha)\mathds{1}_{[\overline{V_\alpha},x_2)}(x)+(1-\alpha)(1+\lambda)(1+\theta)S_\Q(x)\mathds{1}_{[x_2,\hat{x})}(x) \\
        &+\hat{S}_\Q(x)\wedge (1-\alpha)(1+\lambda)(1+\theta)S_\Q(x)\mathds{1}_{[\hat{x},M]}(x).
    \end{align*}
\end{itemize}

Note that the worst-case survival function for Case (\rmnum{1}) is exactly the benchmark survival function. In practice, as $\alpha$ is typically close to 1 (e.g., $99\%$ or $97.5\%$) but $\theta$ is relatively small, the DM would only change her choice of the indemnity function while staying with the benchmark distribution unless the worst-case VaR constraint is sufficiently restrictive. For Cases (\rmnum{2}) and (\rmnum{3}), there might be an upward jump for $G^*(x;\beta)$ at $\overline{V_\alpha}$. Thus, one needs to refer to Theorem \ref{th:beta} for the worst-case survival function. Given the complexity of the results in these theorems, it is nearly impossible to determine $b$ or $b(\beta)$ analytically. Therefore, we compute this parameter and the corresponding Lagrangian multipliers (i.e., $\beta^*_\lambda$ and $\lambda^*$) using numerical methods.  

Our analysis of the three cases reveals that the worst-case survival function tends to exceed the benchmark in the tail region (beyond $\overline{V_\alpha}$) when the constraint on the guaranteed VaR performance becomes effective. This complements the finding of \cite{boonen2024robust} where the DM always stays with the benchmark distribution in the absence of such a constraint. 

Near the end of this section, we present the following numerical example, exploring the effect of $A$ (the acceptable worst-case VaR) on the worst-case survival function as well as the corresponding net price.

\begin{example}
Aligning with Example \ref{example:alpha-maxmin-var}, we assume that the loss follows a truncated exponential distribution with mean 1 and truncation point 100. We use $\varphi(x)=(x+1)\ln(x+1)$ as the Bregman generator. To facilitate computation, this example concerns the effect of $A$ alone instead of the joint effect of $(A,\epsilon)$, where $\epsilon$ denotes the radius of BW ball. For the considered three levels of $A$, the left panel of Figure \ref{fig:TVaR_numerical} displays the worst-case survival functions, where $A=1.406$ corresponds to the case $\lambda=0$. As $A$ decreases, the resulting worst-case survival function becomes larger in the tail region, indicating that the DM would exaggerate more the tail risk. 

The right panel of Figure \ref{fig:TVaR_numerical} displays the net prices for the considered three cases. In the absence of the guaranteed VaR performance, the net price is positive for small and moderate losses and then becomes negative (or non-positive) for large ones, implying that the decision maker optimally selects a stop-loss indemnity. When a guaranteed VaR constraint is imposed, the net price becomes smaller for moderate losses, reflecting the DM's stronger desire to purchase insurance for those losses. In the meantime, the net price becomes zero for tail losses, which, as shown in Theorem \ref{thm:main-prob2}, implies that the coverage for tail losses is chosen to meet the guaranteed VaR constraint.
    
\begin{figure}[!htbp]
\centering
\minipage{0.5\textwidth}
 \centering
  \includegraphics[width=\linewidth]{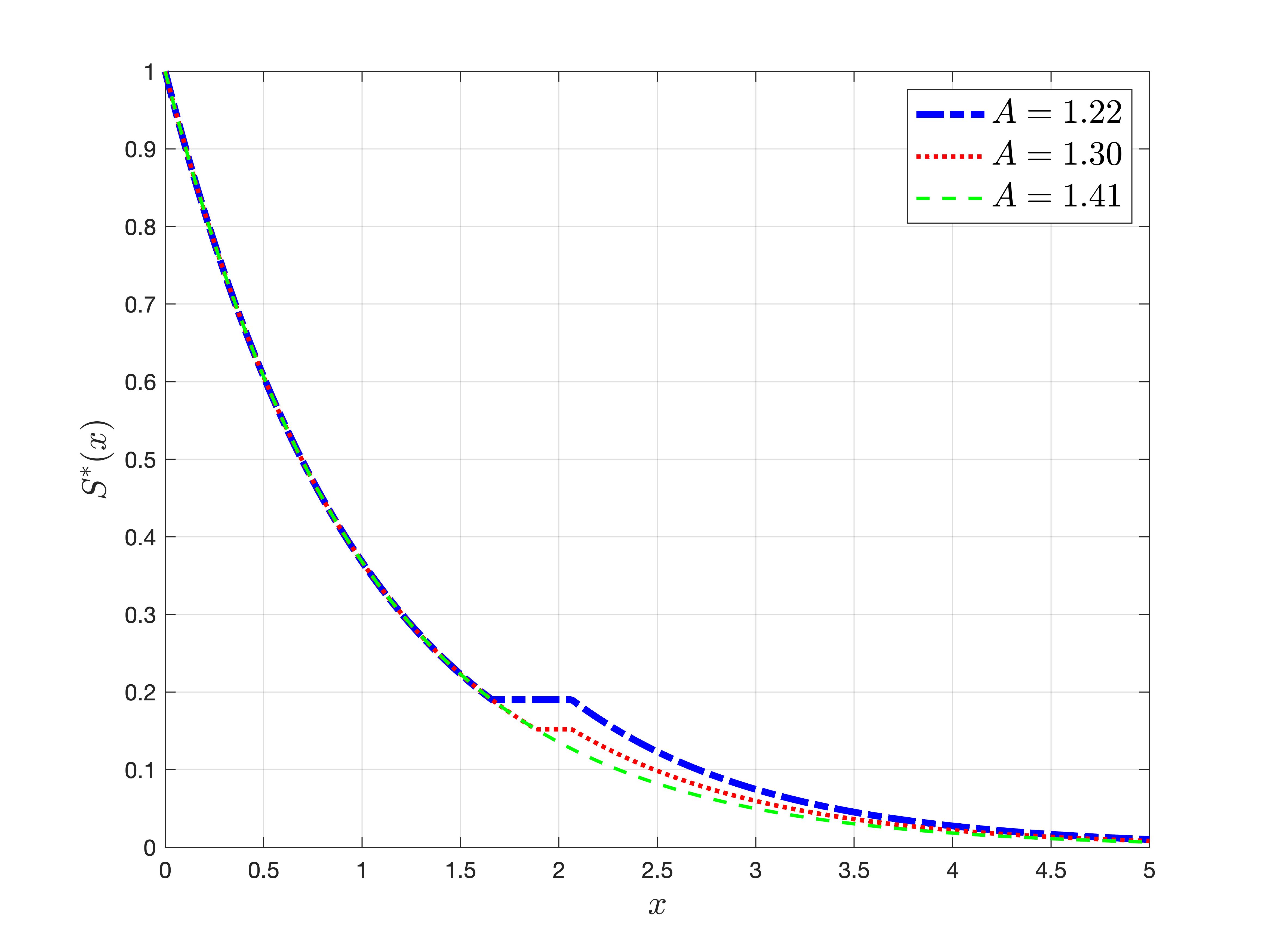}
\endminipage
\minipage{0.5\textwidth}
  \centering
  \includegraphics[width=\linewidth]{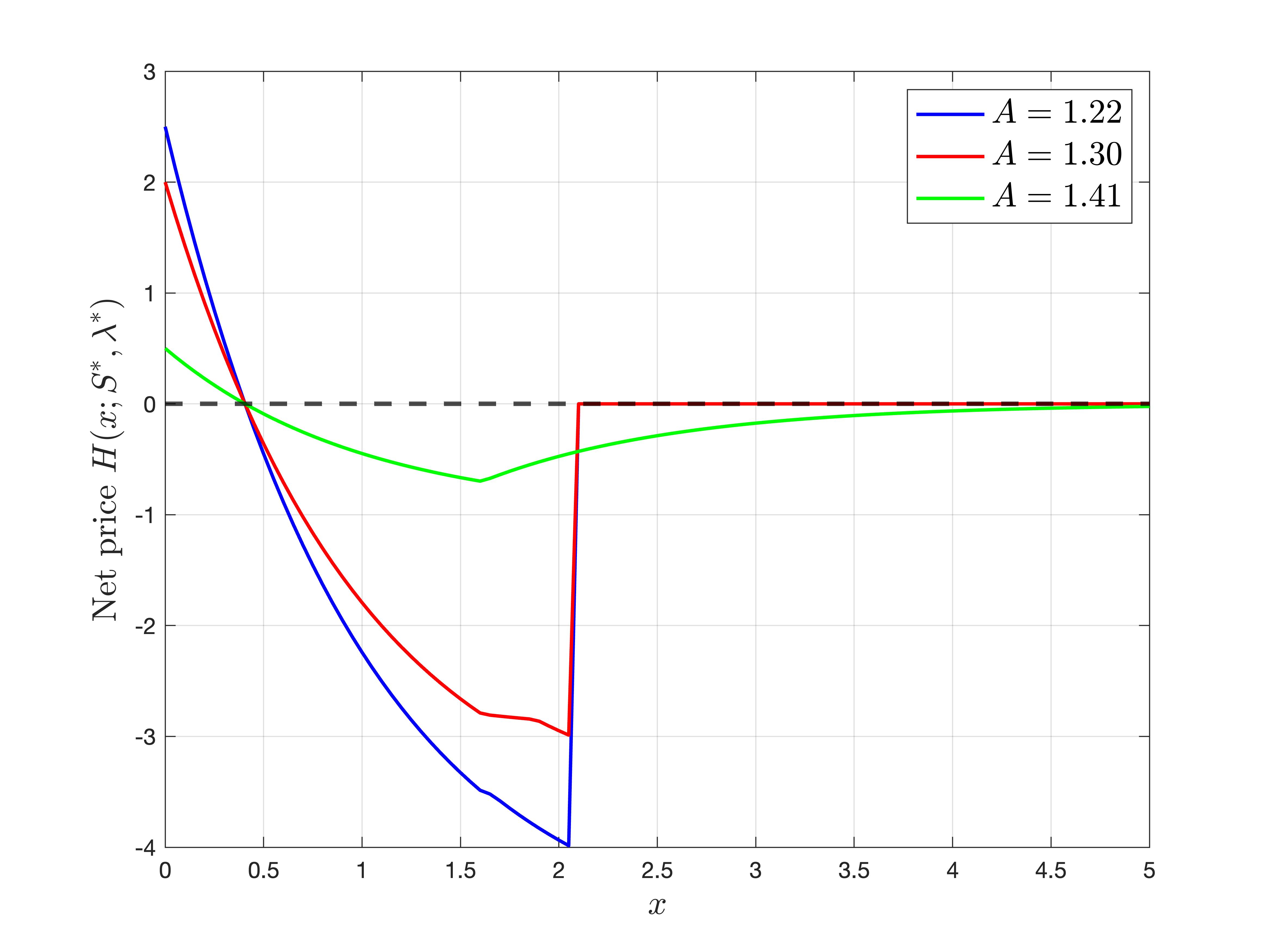}
\endminipage
\caption{The worst-case survival functions and the corresponding net prices.}
\label{fig:TVaR_numerical}
\end{figure}
\end{example}

\section{Concluding remarks}\label{sec:conclusion}
In this paper, we study two optimal insurance problems from a buyer's perspective under distributional uncertainty. Extending the existing literature on distance-based ambiguity sets, we employ a BW ball that allows asymmetric penalization of deviations from the benchmark distribution. In the first problem, we model the DM's ambiguity aversion using the conventional $\alpha$-maxmin model, with VaR serving as the risk measure that captures the DM's risk attitude. We derive the optimal indemnity function for this model in its closed form and analyze how the asymmetry of BW divergence shapes the structure of the indemnity function analytically and numerically. In the second problem, we revisit the worst-case-based insurance problem under the convex distortion risk measure, but subject to a constraint imposed on the worst-case VaR of the DM's terminal loss. We analytically derive the optimal indemnity function and the worst-case loss distribution, and examine their sensitivity to the acceptable worst-case VaR, offering economic interpretations of these relationships. Additional insights are illustrated through a concrete example involving TVaR.

Several promising avenues for future research emerge from this study. First, we note that the methodology for the guaranteed VaR performance model is still applicable when imposing multiple guaranteed VaR constraints, and there will be more parameters to optimize correspondingly. Developing an efficient algorithm for such double-layer optimization would make our model easier to implement in practice. Second, a natural theoretical extension to the guaranteed VaR performance model is to replace VaR with alternative risk measures in the constraint, which introduces significant complexity since the worst-case loss distribution for the constraint may become intertwined with the indemnity choice. Developing a methodology to resolve this issue could enable a broader framework that accommodates multiple risk measures and accounts for preference uncertainty when a finite set of preferences is available. These directions warrant further exploration.

\section*{Acknowledgments}
Wenjun Jiang acknowledges the financial support received from the Natural Sciences and Engineering Research Council of Canada (Grant RGPIN-2020-04204), Alberta Innovates and Canadian Institute of Actuaries. Qingqing Zhang acknowledges the financial support from the National Natural Science Foundation of China (Grant No. 12601924). Yiying Zhang acknowledges the financial support from the National Natural Science Foundation of China (Grant No. 12571513), and Shenzhen Science and Technology Program (Grant No. JCYJ20250604144210013).  The authors are listed in alphabetical order.

\section*{Declarations of interest}
No potential competing or conflict interests were reported by the authors.

\setlength{\bibsep}{0.0pt}
\bibliographystyle{apalike2}
\bibliography{robust}

\appendix
\setcounter{equation}{0}
\renewcommand\theequation{A.\arabic{equation}}
\renewcommand\thedefinition{A.\arabic{definition}}
\renewcommand\thelemma{A.\arabic{lemma}}

\section{Proofs of the main results}
\subsection*{Proof of Lemma \ref{lem:survival-rep-BW}}
\noindent To prove Lemma \ref{lem:survival-rep-BW}, we need the following lemma.
\begin{lemma}\label{lem:EXY}
Let $X$ and $Y$ be real-valued random variables satisfying $\mathbb{E}[|X||Y|] < \infty$. Then
\begin{equation}\label{eq:XY-tail}
\begin{aligned}
\mathbb{E}[XY]
&= \int_0^\infty \int_0^\infty \PP(X>x,\, Y>y)\, dx\, dy
- \int_0^\infty \int_0^\infty \PP(X>x,\, Y<-y)\, dx\, dy\\
&\quad - \int_0^\infty \int_0^\infty \PP(X<-x,\, Y>y)\, dx\, dy
+ \int_0^\infty \int_0^\infty \PP(X<-x,\, Y<-y)\, dx\, dy.
\end{aligned}
\end{equation}
\end{lemma}
\begin{proof}
The positive and negative parts of any random variable $Z$ can be defined as $Z^+=\max\{Z,0\}$ and $Z^-=\max\{-Z,0\}$. Thus, we have $XY=(X^+-X^-)(Y^+-Y^-)=X^+Y^+ - X^+Y^- - X^-Y^+ + X^-Y^-$.

The layer-cake representation gives $Z^+=\int_0^\infty \mathds{1}_{\{Z>z\}}dz$ and $Z^-=\int_0^\infty \mathds{1}_{\{Z<-z\}}dz$, where the equations hold almost surely. Hence, for the term $X^+Y^+$, we have
\begin{equation*}
X^+Y^+=\left(\int_0^\infty \mathds{1}_{\{X>x\}}\,dx\right)\left(\int_0^\infty \mathds{1}_{\{Y>y\}}\,dy\right)=\int_0^\infty \int_0^\infty \mathds{1}_{\{X>x,\, Y>y\}}\,dx\,dy.
\end{equation*}
Applying Tonelli's theorem then leads to
\begin{equation*}
\mathbb{E}[X^+Y^+]=\int_0^\infty \int_0^\infty \E[\mathds{1}_{\{X>x,\, Y>y\}}]\,dx\,dy=\int_0^\infty \int_0^\infty \PP(X>x,\, Y>y)\,dx\,dy.
\end{equation*}
Similarly, one obtains
\begin{align*}
\mathbb{E}[X^+Y^-]&=\int_0^\infty \int_0^\infty \PP(X>x,Y<-y)\,dx\,dy, \\
\mathbb{E}[X^-Y^+]&=\int_0^\infty \int_0^\infty \PP(X<-x,Y>y)\,dx\,dy, \\
\mathbb{E}[X^-Y^-]&=\int_0^\infty \int_0^\infty \PP(X<-x,Y<-y)\,dx\,dy.
\end{align*}
All integrals are well defined due to $\mathbb{E}[|X||Y|]<\infty$. This completes the proof.
\end{proof}

Now we are ready to prove Lemma \ref{lem:survival-rep-BW}. It can be easily verified that
\begin{align*}
\int_0^1 \varphi(F_i^{-1}(t))\,dt
&= \int_0^M \varphi(x)\,dF_i(x) 
= \int_0^M \Big[\varphi(0) + \int_0^x \varphi'(u)\,du \Big] dF_i(x)\\
&= \varphi(0) + \int_0^M \varphi'(u) S_i(u)\, du, \quad i \in \{1,2\},
\end{align*}
where the last equality holds due to Fubini's theorem. Similarly, we have
\begin{equation*}
\int_0^1 \varphi'(F_i^{-1}(t)) F_i^{-1}(t)\,dt= \int_0^M \varphi'(x)xdF_i(x)= \int_0^M \big(\varphi''(u) u + \varphi'(u)\big) S_i(u)\,du, \quad i \in \{1,2\}.
\end{equation*}
For the term $\int_0^1 \varphi'(F_2^{-1}(t)) F_1^{-1}(t)\,dt$, we note that $F_1^{-1}(U)\in [0,M]$ is a non-negative random variable, while $\varphi'(F_2^{-1}(U))$ may be negative. Thus, applying Lemma \ref{lem:EXY} yields
\begin{align} 
\int_0^1 \varphi'(F_2^{-1}(t))F_1^{-1}(t)\,dt=&\mathbb{E}\big[\varphi'(F_2^{-1}(U))F_1^{-1}(U)\big] \nonumber\\
=&\int_0^M\int_0^{\infty}\PP\big(\varphi'(F_2^{-1}(U))>x, F_1^{-1}(U)>y\big)\,dx\,dy \nonumber \\
&-\int_0^M\int_0^{\infty}\PP\big(\varphi'(F_2^{-1}(U))<-x, F_1^{-1}(U)>y\big)\,dx\,dy.  \label{eq:cross-term}
\end{align}

The following lemma gives the survival-function-based representation of $\int_0^1 \varphi'(F_2^{-1}(t)) F_1^{-1}(t)\,dt$.
\begin{lemma}
    Under the assumption of Lemma \ref{lem:survival-rep-BW}, we have 
    \begin{equation*}
  \int_0^1 \varphi'(F_2^{-1}(t))F_1^{-1}(t)\,dt=\varphi'(0)\int_0^M S_1(y)\,dy+\int_0^M\int_0^{M}\varphi''(s)S_2(s)\wedge S_1(y)\,ds\,dy.  
\end{equation*}
\end{lemma}
\begin{proof}
Firstly, if $\varphi'(0)>0$, we have from \eqref{eq:cross-term} that
\begin{align*} 
&\int_0^1 \varphi'(F_2^{-1}(t))F_1^{-1}(t)\,dt\\=&\int_0^M\int_0^{\varphi'(M)}\PP\big(\varphi'(F_2^{-1}(U))>x, F_1^{-1}(U)>y\big)\,dx\,dy\\ =&\int_0^M\int_0^{\varphi'(0)}\PP\big(\varphi'(F_2^{-1}(U))>x, F_1^{-1}(U)>y\big)\,dx\,dy\\ &+\int_0^M\int_{\varphi'(0)}^{\varphi'(M)}\PP\big(\varphi'(F_2^{-1}(U))>x, F_1^{-1}(U)>y\big)\,dx\,dy\\ =&\varphi'(0)\int_0^M\PP\big(F_1^{-1}(U)>y\big)\,dy+\int_0^M\int_{0}^{M}\PP\big(\varphi'(F_2^{-1}(U))>\varphi'(s), F_1^{-1}(U)>y\big)\,d\varphi'(s)\,dy\\=&\varphi'(0)\int_0^M S_1(y)\,dy+\int_0^M\int_0^{M}\PP\big(F_2^{-1}(U)>s, F_1^{-1}(U)>y\big)\varphi''(s)\,ds\,dy\\ =&\varphi'(0)\int_0^M S_1(y)\,dy+\int_0^M\int_0^{M}\PP\big(U>F_2(s), U>F_1(y)\big)\varphi''(s)\,ds\,dy\\ =&\varphi'(0)\int_0^M S_1(y)\,dy+\int_0^M\int_0^{M}\varphi''(s)S_2(s)\wedge S_1(y)\,ds\,dy. 
\end{align*} 
If there exists a $x_0\in [0,M]$ such that $\varphi'(x_0)=0$, then 
\begin{align*}
	&\int_0^1 \varphi'(F_2^{-1}(t))F_1^{-1}(t)\,dt\\=&\int_0^M\int_{\varphi'(x_0)}^{\varphi'(M)}\PP\big(\varphi'(F_2^{-1}(U))>x, F_1^{-1}(U)>y\big)\,dx\,dy\\&-\int_0^M\int_{\varphi'(x_0)}^{-\varphi'(0)}\PP\big(\varphi'(F_2^{-1}(U))<-x, F_1^{-1}(U)>y\big)\,dx\,dy\\
    =&\int_0^M\int_{x_0}^M\PP\big(\varphi'(F_2^{-1}(U))>\varphi'(s), F_1^{-1}(U)>y\big)d\varphi'(s)\,dy\\
	&-\int_0^M\int_0^{x_0}\PP\big(\varphi'(F_2^{-1}(U))<\varphi'(s), F_1^{-1}(U)>y\big)d\varphi'(s)\,dy\\
    =&\int_0^M\int_{x_0}^M\PP\big(F_2^{-1}(U)>s, F_1^{-1}(U))>y\big)\varphi''(s)\,ds\,dy\\
	&-\int_0^M\int_0^{x_0}\PP\big(F_2^{-1}(U)<s, F_1^{-1}(U)>y\big)\varphi''(s)\,ds\,dy\\
    =&\int_0^M\int_{x_0}^M\varphi''(s)S_2(s)\wedge S_1(y)\,ds\,dy-\int_0^M\int_0^{x_0}\varphi''(s)[S_1(y)-S_2(s)\wedge S_1(y)]\,ds\,dy\\
    =&\int_0^M\int_0^M \varphi''(s)S_2(s)\wedge S_1(y)\,ds\,dy-\int_0^M\int_0^{x_0}\varphi''(s)S_1(y)\,ds\,dy\\
    =&\varphi'(0)\int_0^M S_1(y)\,dy+\int_0^M\int_0^{M}\varphi''(s)S_2(s)\wedge S_1(y)\,ds\,dy.  
\end{align*}
If $\varphi'(M)<0$, then
\begin{align*}
	&\int_0^1 \varphi'(F_2^{-1}(t))F_1^{-1}(t)\,dt\\=&-\int_0^M\int_0^{-\varphi'(0)}\PP\big(\varphi'(F_2^{-1}(U))<-x, F_1^{-1}(U)>y\big)\,dx\,dy\\
    =&-\int_0^M\int_0^{-\varphi'(M)}\PP\big(\varphi'(F_2^{-1}(U))<-x, F_1^{-1}(U)>y\big)\,dx\,dy\\
    &-\int_0^M\int_{-\varphi'(M)}^{-\varphi'(0)}\PP\big(\varphi'(F_2^{-1}(U))<-x, F_1^{-1}(U)>y\big)\,dx\,dy\\
    =&\varphi'(M)\int_0^M\PP\big(F_1^{-1}(U)>y\big)\,dy-\int_0^M\int_0^M\PP\big(\varphi'(F_2^{-1}(U))<\varphi'(s), F_1^{-1}(U)>y\big)d\varphi'(s)\,dy\\
    =&\varphi'(M)\int_0^M S_1(y)\,dy-\int_0^M\int_0^M\PP\big(F_2^{-1}(U)<s, F_1^{-1}(U)>y\big)\varphi''(s)\,ds\,dy\\
    =&\varphi'(M)\int_0^M S_1(y)\,dy-\int_0^M\int_0^M\varphi''(s)[S_1(y)-S_2(s)\wedge S_1(y)]\,ds\,dy\\
    =&\Big[\varphi'(M)-\int_0^M\varphi''(s)\,ds\Big]\int_0^M S_1(y)\,dy+\int_0^M\int_0^{M}\varphi''(s)S_2(s)\wedge S_1(y)\,ds\,dy\\
    =&\varphi'(0)\int_0^M S_1(y)\,dy+\int_0^M\int_0^{M}\varphi''(s)S_2(s)\wedge S_1(y)\,ds\,dy. 
    \end{align*}
    This ends the proof.
\end{proof}

The survival-function-based representation of $\mathscr{B}_\varphi[F_1,F_2]$ in Lemma \ref{lem:survival-rep-BW} can then be derived by collecting the above results. The proof is now complete. \qed

\subsection*{Proof of Proposition \ref{th4-1}}
Before proving Proposition \ref{th4-1}, we present the following lemma.
\begin{lemma}
For any $x,y\in \mathbb{R}$, it holds that $B_{\varphi}(x,y)\ge 0$. Moreover, for any fixed $y\in \mathbb{R}$, the mapping $x\mapsto B_{\varphi}(x,y)$ is decreasing on $(-\infty,y]$ and increasing on $[y,\infty)$.
\end{lemma}
\begin{proof}
Since $\varphi$ is convex and continuously differentiable, we have
\[
\varphi(x) \ge \varphi(y) + \varphi'(y)(x - y),
\qquad \forall\, x,y\in\mathbb{R},
\]
which immediately implies
\[
B_{\varphi}(x,y)=\varphi(x)- \varphi(y) - \varphi'(y)(x - y)\ge 0.
\]

Next, for fixed $y\in\mathbb{R}$, we first show that the function $x\mapsto B_{\varphi}(x,y)$ is increasing on $[y,\infty)$.  
Let $x_1>x_2\ge y$. By the Mean Value Theorem, there exists some $x_0\in(x_2,x_1)$ such that
\begin{align*}
B_{\varphi}(x_1,y)-B_{\varphi}(x_2,y)
&=\varphi(x_1)-\varphi(x_2)-\varphi'(y)(x_1-x_2) \\
&=[\varphi'(x_0)-\varphi'(y)](x_1-x_2).
\end{align*}
Since $\varphi$ is strictly convex, its derivative $\varphi'$ is increasing, which implies $\varphi'(x_0)> \varphi'(y)$. Hence,
\[
B_{\varphi}(x_1,y)-B_{\varphi}(x_2,y)> 0.
\]
An analogous argument applies to the interval $(-\infty,y]$. This completes the proof.
\end{proof}

Now we are ready to prove Proposition \ref{th4-1}. First, if $\epsilon > \int_{\alpha}^{1} B_{\varphi}(M,F_0^{-1}(t)) \, dt$, we prove that $M=\sup_{F \in \mathcal{B}(F_0,\epsilon)} \mathrm{VaR}^F_{\alpha}(X)$.
Since $X$ is bounded above by~$M$, it implies
\[
\sup_{F\in\mathcal{B}(F_0,\epsilon)} \mathrm{VaR}^F_{\alpha}(X)\le M.
\]
We show that $M$ is the supremum.

Let $\delta_{1}\in(0, M-F_0^{-1}(\alpha)]$. By the monotonicity of $x\mapsto B_{\varphi}(x,y)$, we obtain
\begin{align*}
\int_{\alpha}^{F_0(M-\frac{\delta_{1}}{2})}
B_{\varphi}\!\left(M-\tfrac{\delta_{1}}{2},F_0^{-1}(t)\right)\! dt
&\le \int_{\alpha}^{F_0(M-\frac{\delta_{1}}{2})}B_{\varphi}(M,F_0^{-1}(t))\, dt \\
&\le \int_{\alpha}^{1}B_{\varphi}(M,F_0^{-1}(t))\, dt
<\epsilon.
\end{align*}
Define $\hat{F}$ via
\begin{equation}\label{eqhatF}
\hat{F}^{-1}(t)=
\begin{cases}
M-\frac{\delta_{1}}{2}, & t\in(\alpha-\xi_{1}, F_0(M-\frac{\delta_{1}}{2})], \\
F_0^{-1}(t), & t\notin(\alpha-\xi_{1}, F_0(M-\frac{\delta_{1}}{2})],
\end{cases}    
\end{equation}
where $\xi_{1}$ is chosen such that 
\begin{align*} &\int_{0}^{1}B_{\varphi}(\hat{F}^{-1}(t),F_0^{-1}(t))\,dt=\int_{\alpha-\xi_{1}}^{F_0(M-\frac{\delta_{1}}{2})}B_{\varphi}(M-\frac{\delta_{1}}{2},F_0^{-1}(t))\ dt \le\epsilon. \end{align*}
Then $\hat{F}\in\mathcal{B}(F_0,\epsilon)$ and
\(
\mathrm{VaR}_{\alpha}^{\hat{F}}(X)>M-\delta_{1}.
\)
Since $\delta_{1}>0$ is arbitrary, the supremum equals $M$.

Next, we prove Equation $\eqref{eqsup}$. The right-continuity of $F_0$ implies the right-continuity of
\[
D\mapsto \int_{\alpha}^{F_0(D)}B_{\varphi}(D,F_0^{-1}(t))\, dt,
\]
and thus by the definition of $\overline{V_\alpha}$, we have
\[
\int_{\alpha}^{F_0(\overline{V_\alpha})}B_{\varphi}(\overline{V_\alpha},F_0^{-1}(t))\, dt\ge \epsilon.
\]

We first verify that $\overline{V_\alpha}$ is an upper bound.  
If not, suppose $\tilde{F}\in\mathcal{B}(F_0,\epsilon)$ satisfies $\tilde{F}^{-1}(\alpha)>\overline{V_\alpha}$.  
Then $\tilde{F}^{-1}(t)>\overline{V_\alpha}$ for $t\in(\alpha,F_0(\overline{V_\alpha})]$, and hence
\begin{align*}
\int_{0}^{1}B_{\varphi}(\tilde{F}^{-1}(t),F_0^{-1}(t))\, dt
&\ge \int_{\alpha}^{F_0(\overline{V_\alpha})}B_{\varphi}(\tilde{F}^{-1}(t),F_0^{-1}(t))\, dt \\
&> \int_{\alpha}^{F_0(\overline{V_\alpha})}B_{\varphi}(\overline{V_\alpha},F_0^{-1}(t))\, dt 
\ge \epsilon,
\end{align*}
contradicting $\tilde{F}\in\mathcal{B}(F_0,\epsilon)$.

To see that the supremum cannot be smaller than~$\overline{V_\alpha}$, we take any $\delta_{2}\in(0,\overline{V_\alpha}-F_0^{-1}(\alpha)]$. By the definition of~$\overline{V_\alpha}$,
\[
\int_{\alpha}^{F_0(\overline{V_\alpha}-\frac{\delta_{2}}{2})}
B_{\varphi}\!\left(\overline{V_\alpha}-\tfrac{\delta_{2}}{2},F_0^{-1}(t)\right)\! dt 
<\epsilon.
\]
Constructing $\hat{F}$ analogously to $\eqref{eqhatF}$ yields $\hat{F}\in\mathcal{B}(F_0,\epsilon)$ and $\mathrm{VaR}_{\alpha}^{\hat{F}}(X)>\overline{V_\alpha}-\delta_{2}$, which completes the proof of~\eqref{eqsup}.

Finally, we demonstrate the nonexistence of $F^{*}\in\mathcal{B}(F_0,\epsilon)$ satisfying  \begin{equation*}
\mathrm{VaR}^{F^*}_{\alpha}(X)=\sup_{F\in\mathcal{B}(F_0,\epsilon)}\mathrm{VaR}^F_{\alpha}(X).
 \end{equation*}
 The proof is provided specifically for the case where $\epsilon\leq\int_{\alpha}^{1}B_{\varphi}(M,F_0^{-1}(t))\, dt$.

Assume that $F^{*}\in\mathcal{B}(F_0,\epsilon)$ satisfies
\(
F^{*-1}(\alpha)=\overline{V_\alpha}.
\)
For any $\delta_{3}\in(0,\overline{V_\alpha}-F_0^{-1}(\alpha))$ there exists $\xi_{3}\in(0,\alpha)$ such that
\(
F^{*-1}(t)\in(\overline{V_\alpha}-\delta_{3},\overline{V_\alpha}]
\)
for $t\in(\alpha-\xi_{3},\alpha]$.  
Then
\begin{align*}
\int_{0}^{1}B_{\varphi}(F^{*-1}(t),F_0^{-1}(t))\, dt 
&\ge \int_{\alpha-\xi_{3}}^{\alpha}
B_{\varphi}(\overline{V_\alpha}-\delta_{3},F_0^{-1}(t))\, dt \\
&\quad + \int_{\alpha}^{F_0(\overline{V_\alpha})}
B_{\varphi}(\overline{V_\alpha},F_0^{-1}(t))\, dt
> \epsilon,
\end{align*}
contradicting $F^{*}\in\mathcal{B}(F_0,\epsilon)$.  
Thus, the supremum is not attained. The proof is complete. \qed

\subsection*{Proof of Proposition \ref{th4-2}} 
We first show $\underline{V_\alpha}\ge 
\inf_{F \in \mathcal{B}(F_0,\epsilon)} \mathrm{VaR}^F_{\alpha}(X)$. Define the distribution function $\hat{F}$ via
\begin{equation}\label{eqhatF2}
\hat{F}^{-1}(t)=
\begin{cases}
\underline{V_\alpha}, & t\in(F_0(\underline{V_\alpha}),\alpha], \\
F_0^{-1}(t), & t\notin (F_0(\underline{V_\alpha}),\alpha].
\end{cases}    
\end{equation}
For $t\le F_0(\underline{V_\alpha})$, we have
$\hat{F}^{-1}(t)=F_0^{-1}(t)\le \underline{V_\alpha}$, while for 
$t\in(\alpha,1]$, 
$\hat{F}^{-1}(t)=F_0^{-1}(t)\ge F_0^{-1}(\alpha)\ge \underline{V_\alpha}$.  
Hence, $\hat{F}^{-1}$ is a valid quantile function, which induces a well-defined CDF $\hat{F}$.

By the definition of $\underline{V_\alpha}$, we have
\[
\int_{F_0(\underline{V_\alpha})}^{\alpha}
B_{\varphi}(\underline{V_\alpha},F_0^{-1}(t))\, dt\le \epsilon.
\]
Consequently,
\begin{align*}
\int_{0}^{1}B_{\varphi}(\hat{F}^{-1}(t),F_0^{-1}(t))\, dt
&= \int_{F_0(\underline{V_\alpha})}^{\alpha}
     B_{\varphi}(\underline{V_\alpha},F_0^{-1}(t))\, dt
   \le \epsilon,
\end{align*}
which shows that $\hat{F}\in\mathcal{B}(F_0,\epsilon)$.  
Therefore, we have $\underline{V_\alpha}\ge 
\inf_{F\in\mathcal{B}(F_0,\epsilon)}\mathrm{VaR}^F_\alpha(X)$.

To prove $\underline{V_\alpha}\le 
\inf_{F\in\mathcal{B}(F_0,\epsilon)}\mathrm{VaR}^F_\alpha(X)$, we show that for any $F\in\mathcal{B}(F_0,\epsilon)$, $\underline{V_\alpha}\le 
\mathrm{VaR}_\alpha^F(X)$. If $F^{-1}(\alpha)\ge F_0^{-1}(\alpha)$, it follows that $\underline{V_\alpha}\le F_0^{-1}(\alpha)\le F^{-1}(\alpha)=\mathrm{VaR}_\alpha^F(X)$. If $F^{-1}(\alpha)<F_0^{-1}(\alpha)$, we need to prove
\begin{equation*}
    \int_{F_0(F^{-1}(\alpha))}^{\alpha}
       B_{\varphi}(F^{-1}(\alpha),F_0^{-1}(t))\, dt\le \epsilon.
\end{equation*}
For any $t\in [F_0(F^{-1}(\alpha)),\alpha]$, we have $F^{-1}(t)\leq F^{-1}(\alpha)\le F_0^{-1}(t)$, so
\begin{align*}
  \int_{F_0(F^{-1}(\alpha))}^{\alpha}
       B_{\varphi}(F^{-1}(\alpha),F_0^{-1}(t))\, dt&\le \int_{F_0(F^{-1}(\alpha))}^{\alpha}
       B_{\varphi}(F^{-1}(t),F_0^{-1}(t))\, dt\\
       &\le\int_0^1
       B_{\varphi}(F^{-1}(t),F_0^{-1}(t))\, dt\le \epsilon. 
\end{align*}
This completes the proof.\qed

\subsection*{Proof of Lemma \ref{thm:sol-exist-prob2}}
Given the finite support $[0,M]$, the set $\mathcal{I}$ is compact under the $L^\infty$ metric. Let $\mathcal{T}(I)=\pi(I(X))-I(\overline{V_\alpha})$, it is apparent that $\mathcal{T}$ is continuous in $I$ under the $L^\infty$ metric. Under Assumption \ref{Assum:non-trivial}, the pre-image $\mathcal{T}^{-1}((-\infty, A-\overline{V_\alpha}]):=\left\{I\in\mathcal{I}: \mathcal{T}(I)\le A-\overline{V_\alpha}\right\}$ is closed. Hence, $\mathcal{T}^{-1}((-\infty,A-\overline{V_\alpha}])$ is also compact. 

Denote by
$$\mathcal{K}(I):=\sup_{F\in\mathcal{B}(F_0,\epsilon)}\rho_g^F(X-I(X)+\pi(I(X))).$$ 
Given $\{I_n\}\to I$ within $\mathcal{T}^{-1}((-\infty,A-\overline{V_\alpha}])$, for any $\widetilde{\epsilon}>0$, there exists a $F_{\widetilde{\epsilon}}\in\mathcal{B}(F_0,\epsilon)$ such that 
\begin{equation*}
 \mathcal{K}(I)-\widetilde{\epsilon}\le \rho_g^{F_{\widetilde{\epsilon}}}(X-I(X)+\pi(I(X))).   
\end{equation*}
By Theorem 2.1 of \cite{wang2020distortion} we have that $\rho_g^F(X-I(X)+\pi(I(X)))$ is continuous in $I$ under the $L^\infty$ metric. As such, 
\begin{equation*}
\rho_g^{F_{\widetilde{\epsilon}}}(X-I_n(X)+\pi(I_n(X)))\to \rho_g^{F_{\widetilde{\epsilon}}}(X-I(X)+\pi(I(X))).
\end{equation*}
Therefore,
\begin{align*}
    \liminf_{n\to\infty}\mathcal{K}(I_n)\ge&\ \liminf_{n\to\infty}\rho_g^{F_{\widetilde{\epsilon}}}(X-I_n(X)+\pi(I_n(X))) \\
    =&\ \rho_g^{F_{\widetilde{\epsilon}}}(X-I(X)+\pi(I(X))) \\
    \ge&\ \mathcal{K}(I)-\widetilde{\epsilon}.
\end{align*}
Given that $\widetilde{\epsilon}$ is arbitrary, we have $\liminf_{n\to\infty}\mathcal{K}(I_n)\ge\mathcal{K}(I)$. Hence, $\mathcal{K}(I)$ is lower-semicontinuous over $\mathcal{T}^{-1}((-\infty,A-\overline{V_\alpha}])$. Thus, the minimum of $\mathcal{K}(I)$ is attainable over the compact set $\mathcal{T}^{-1}((-\infty,A-\overline{V_\alpha}])$. The proof is complete.  \qed

\subsection*{Proof of Lemma \ref{thm:sol-exist-inner}}
Given the finite support $[0,M]$, by using the Prokhorov's theorem (see Theorem 5.1 of \cite{billingsley2013convergence}) we get that the set of CDFs on $[0,M]$ is compact under the topology of weak convergence. Denote the objective function of \eqref{Prob:inner} by $\mathcal{J}(S)$. Given that $g$ is a continuous concave distortion function, we have $g(S_n(x))\to g(S(x))$ if $S_n(x)\to S(x)$ at any continuity point $x\in[0,M]$. By using the Lebesgue Dominated Convergence Theorem, we get $\mathcal{J}(S_n)\to \mathcal{J}(S)$. Thus, $\mathcal{J}(S)$ is continuous under pointwise convergence, or weak convergence. Since $S_0$ is a strictly feasible solution, the constraint 
\begin{equation*}
\int_0^M\left([\varphi'(x)-\varphi'(0)]S(x)-\int_0^M\varphi''(y)S_0(y)\wedge S(x)\,dy\right)\,dx\le\zeta
\end{equation*}
generates a compact feasible domain of $S$, over which the maximum of $\mathcal{J}(S)$ is attainable. This ends the proof. \qed

\subsection*{Proof of Proposition \ref{prop:beta>0}}
By the definitions of $x_1$ and $x_2$, we have
\begin{align*} &g(G(x)) \le (1+\lambda)(1+\theta)S_{\mathbb Q}(x)-\lambda, \quad \forall x\in [0,x_1),\\ &g(G(x)) \le (1+\lambda)(1+\theta)S_{\mathbb Q}(x), \quad \forall x\in [\overline{V_\alpha},x_2).
\end{align*}
We study the cases $\beta=0$ and $\beta>0$ separately.

\medskip
\noindent\textbf{\underline{Case 1: $\beta=0$}}

In this case, the objective function of Problem~\eqref{eq:totalLg-e} over $\mathcal{G}$ can be written as
\begin{equation}\label{eq:th1prob} \begin{aligned} \sup_{G\in \mathcal{G}} \Bigg\{&\int_0^{x_1}\! g(G(x)) + \int_{x_1}^{\overline{V_\alpha}} \!g(G(x))\wedge\left((1+\lambda)(1+\theta)S_\Q(x)-\lambda\right)\\ &+\int_{\overline{V_\alpha}}^{x_2}\! g(G(x)) + \int_{x_2}^M \! g(G(x))\wedge(1+\lambda)(1+\theta)S_\Q(x)\Bigg\}. \end{aligned} \end{equation}
  We first solve Problem~\eqref{eq:th1prob} and denote its solution set by $\tilde{\mathcal{G}}$. Then, we select the one within $\tilde{\mathcal{G}}$ that minimizes
  \begin{equation}\label{eq:auxiliary}  
  \int_0^M \Big(\big[\varphi'(x)-\varphi'(0)\big]G(x)-\int_0^{M}\varphi''(y)S_0(y)\wedge G(x)\,dy\Big)\,dx,
  \end{equation}
  thereby identifying the smallest BW ball for which the constraint is non-binding.

\medskip
\noindent\textbf{\underline{Case 1(a): $x\in[0,x_1)\cup [\overline{V_\alpha},x_2)$}}

In this region, since $g$ is increasing and concave, any $t\in[g^{-1}(1),1]$ maximizes $g(t)$ over $[0,1]$. Hence, $\tilde{\mathcal{G}}=[g^{-1}(1),1]$. To make the BW ball as small as possible, $G^*(x;0)=g^{-1}(1)\vee S_0(x)$.

\medskip
\noindent\textbf{\underline{Case 1(b): $x\in[x_1,\overline{V_\alpha})\cup [x_2,M]$}}

For $x\in\mathscr{A}_1$, we have
\begin{equation*}
S_0(x)\le g^{-1}\big((1+\lambda)(1+\theta)S_{\mathbb Q}(x)-\lambda\big).
\end{equation*}
Since
\begin{equation*}
g(G(x))\wedge\big((1+\lambda)(1+\theta)S_{\mathbb Q}(x)-\lambda\big)
\le (1+\lambda)(1+\theta)S_{\mathbb Q}(x)-\lambda,
\end{equation*}
any maximizer must satisfy $G(x)\ge g^{-1}((1+\lambda)(1+\theta)S_{\mathbb Q}(x)-\lambda)$. To minimize \eqref{eq:auxiliary}, we have
\begin{equation*}
G^*(x;0)=g^{-1}\big((1+\lambda)(1+\theta)S_{\mathbb Q}(x)-\lambda\big).
\end{equation*}

For $x\in\mathscr{B}_1$, we note that
\begin{equation*}
g(G(x))\wedge\big((1+\lambda)(1+\theta)S_{\mathbb Q}(x)-\lambda\big)
\le g(G(x))\wedge g(S_0(x))\le g(S_0(x)),
\end{equation*}
which implies that the choice $G^*(x;0)=S_0(x)$ is optimal. 

By similar arguments, we can derive 
\begin{equation*}
G^*(x;0)=
\begin{cases}
g^{-1}\big((1+\lambda)(1+\theta)S_{\mathbb Q}(x)\big), & x\in \mathscr{A}_2,\\
S_0(x), & x\in \mathscr{B}_2.
\end{cases}
\end{equation*}
To this end, $G^*(x;0)$ is fully characterized.

\medskip
\noindent\textbf{\underline{Case 2: $\beta>0$}}

Based on the definitions of $x_1$ and  $x_2$, the objective function of $\eqref{eq:totalLg-e}$ can be written as
\begin{equation*}
\begin{aligned}
&\sup_{G\in \mathcal{G}} \Bigg\{
\int_0^{x_1}\!\Big\{ g(G(x)) - \beta\!\Big(\big[\varphi'(x)-\varphi'(0)\big]G(x)-\int_0^{M}\varphi''(y)S_0(y)\wedge  G(x)\,dy\Big)
\Big\}\,dx \\
&\quad + 
\int_{x_1}^{\overline{V_\alpha}} \!\Big\{ g(G(x))\wedge\left((1+\lambda)(1+\theta)S_\Q(x)-\lambda\right)-\beta
\Big(\big[\varphi'(x)-\varphi'(0)\big]G(x)\\
&\quad -\int_0^{M}\varphi''(y)S_0(y)\wedge  G(x)\,dy\Big)
\Big\}\,dx\\
&\quad+\int_{\overline{V_\alpha}}^{x_2}\!\Big\{ g(G(x)) - \beta\!\Big(\big[\varphi'(x)-\varphi'(0)\big]G(x)-\int_0^{M}\varphi''(y)S_0(y)\wedge  G(x)\,dy\Big)
\Big\}\,dx \\
&\quad + 
\int_{x_2}^M \!\Big\{ g(G(x))\wedge(1+\lambda)(1+\theta)S_\Q(x)-\beta
\Big(\big[\varphi'(x)-\varphi'(0)\big]G(x)\\
&\quad -\int_0^{M}\varphi''(y)S_0(y)\wedge  G(x)\,dy\Big)
\Big\}\,dx \Bigg\}.
\end{aligned}
\end{equation*}
The following functions are defined to facilitate later discussions:
\begin{equation*}
\begin{aligned}
&K_1(t;x) = g(t)-\beta\!\Big(\big[\varphi'(x)-\varphi'(0)\big]t-\int_0^{M}\varphi''(y)S_0(y)\wedge  t\,dy\Big),\\
&K_2(t;x) = (1+\lambda)(1+\theta)S_\Q(x)-\lambda- \beta\!\Big(\big[\varphi'(x)-\varphi'(0)\big]t-\int_0^{M}\varphi''(y)S_0(y)\wedge  t\,dy\Big),\\
&K_3(t;x) = g(t)\wedge\left((1+\lambda)(1+\theta)S_\Q(x)-\lambda\right)-\beta\!\Big(\big[\varphi'(x)-\varphi'(0)\big]t-\int_0^{M}\varphi''(y)S_0(y)\wedge  t\,dy\Big).
\end{aligned}
\end{equation*}

\medskip
\noindent\textbf{\underline{Case 2(a): $x\in[0,x_1)\cup [\overline{V_\alpha},x_2)$}}

In this case, since $g(G(x))\le (1+\lambda)(1+\theta)S_{\mathbb Q}(x)-\lambda$, the maximization problem reduces to
\begin{equation}\label{eqsupK1}
\sup_{t\in[S_0(x),1]} K_1(t;x),
\end{equation}
where $K_1(t;x)$ is concave due to the concavity of $g$. Since $\{S_0(y)\le t\}=\{y\geq F_0^{-1}(1-t)\}$, $K_1(t;x)$ can be further written as 
\begin{equation*}
K_1(t;x) = g(t) -\beta\big[\varphi'(x)-\varphi'(0)\big]t+\beta\Big(\int_{F_0^{-1}(1-t)}^{M}\varphi''(y)S_0(y)\,dy+\int_0^{F_0^{-1}(1-t)}\varphi''(y) t \,dy\Big),
\end{equation*}
whose derivative exists almost everywhere and decreases over $[S_0(x),1]$. Therefore, the maximizer of $\eqref{eqsupK1}$ is 
\begin{equation*}
\hat G(x;\beta) = \inf\{t\in[S_0(x),1]: K_1'(t;x)\le0\}.
\end{equation*}
Moreover, for any $t\in [0,1]$, we have $F_0(F_0^{-1}(1-t))\geq 1-t$ and $(F_0^{-1})'(1-t)=0$ when 
$F_0(F_0^{-1}(1-t))> 1-t$, it follows that for any $t_0\in [S_0(x),1]$,
\begin{equation*}
K_1'(t_0^+;x)=\lim_{t\to t_0^+}K_1'(t;x) = g'(t_0^+) -\beta\big[\varphi'(x)-\varphi'(F_0^{-1}(1-t_0^+))\big],
\end{equation*}
which is decreasing with $x$ since $\varphi(x)$ is strictly convex. For any $0\leq x\leq y<x_1$ or $\overline{V_\alpha}\leq x\leq y<x_2$, we have $K_1'(t_0^+;x)\geq K_1'(t_0^+;y)$, and thus $\hat G(x;\beta)\geq \hat G(y;\beta)$. Therefore, we deduce $G^*(x;\beta)=\hat G(x;\beta)$.

\medskip
\noindent\textbf{\underline{Case 2(b): $x\in[x_1,\overline{V_\alpha})\cup[x_2,M]$}}

In this case, we define
\begin{equation*}
t_1^*(x) = g^{-1}\left((1+\lambda)(1+\theta)S_\Q(x)-\lambda\right).
\end{equation*}
Then the function $K_3(t;x)$ can be written in pieces as
\begin{equation*}
K_3(t;x) = 
\begin{cases}
K_1(t;x), & t\le t_1^*(x),\\[0.4em]
K_2(t;x), & t>t_1^*(x).
\end{cases}
\end{equation*}
It follows that
	\begin{align*}
		&K_3'({t_1^*(x)}^-;x)=g'({t_1^*(x)}^-)-\beta\big[\varphi'(x) -\varphi'(F_0^{-1}(1-{t_1^*(x)}^-))\big], \\
		&K_3'({t_1^*(x)}^+;x)=-\beta\big[\varphi'(x) -\varphi'(F_0^{-1}(1-{t_1^*(x)}^+))\big].
	\end{align*}
	On $\mathscr{A}_1$, we have 
	\begin{align*}
		&g(t_1^*(x))=(1+\lambda)(1+\theta)S_\Q(x)-\lambda\geq g(S_0(x))\Longrightarrow\ t_1^*(x)\geq S_0(x)\\
	&\Longrightarrow\ F_0(x)\geq 1-t_1^*(x) \Longrightarrow\ x\geq F^{-1}_0(1-t_1^*(x)).
	\end{align*}
    Since $\varphi'(x)$ is an increasing function, we have
	\begin{align*}
		K_3'({t_1^*(x)}^+;x)=&-\beta\big[\varphi'(x) -\varphi'(F_0^{-1}(1-{t_1^*(x)}^+))\big]\le 0,
	\end{align*}
	Hence, on $\mathscr{A}_1$, the maximum of $K_3(t;x)$ over $[S_0(x),1]$ is attained within $[S_0(x),t_1^*(x)]$. Note that the maximum of $K_1(t;x)$ over $[S_0(x),1]$ is attained at $t=\hat{G}(x;\beta)$, so the worst-case function $G^*(x;\beta)=\hat{G}(x;\beta)\wedge t_1^*(x)$ for $x\in\mathscr{A}_1$. 
	
	Similarly, on $\mathscr{B}_1$, we have $t_1^*(x)<S_0(x)$, which leads to $K_3(t;x)=K_2(t;x)$ for $t\in[S_0(x),1]$. Since
    \begin{equation*}
      K_2'(t^+;x)= -\beta\big[\varphi'(x) -\varphi'(F_0^{-1}(1-t^+))\big]\le 0  
    \end{equation*}
	for any $t\in[S_0(x),1]$, the maximum of $K_3(t;x)$ is attained at $S_0(x)$. An analogous argument applies on $[x_2,M]$, leading to
\begin{equation*}
G^*(x;\beta)=
\begin{cases}
\hat G(x;\beta)\wedge g^{-1}\big((1+\lambda)(1+\theta)S_{\mathbb Q}(x)\big), & x\in \mathscr{A}_2,\\
S_0(x), & x\in \mathscr{B}_2.
\end{cases}
\end{equation*}
Combining all cases yields the stated expression for $G^*(x;\beta)$.\qed

\subsection*{Proof of Theorem \ref{th:beta}}
Recall that $G^*$ is the pointwise maximizer of Problem~\eqref{eq:totalLg} over the relaxed admissible class $\mathcal{G}$. Intuitively, if there exists a $b(\beta)$, such that $S(x;b(\beta))$ is pointwise closer to $G^*(x)$ than $
S(x)$, then replacing $S(x)$ by $S(x;b(\beta))$ improves the objective. More precisely, it suffices to verify that for all $x\in[0,M]$ either
\begin{equation*}
 S(x)\le S(x;b(\beta))\le G^*(x)
\quad\text{or}\quad
S(x)\ge S(x;b(\beta))\ge G^*(x).   
\end{equation*}
This property is immediate for  $x\in[0,a_1(\beta))\cup[a_2(\beta),M]$. For $x\in[a_1(\beta),a_2(\beta))$, we look into the following cases. 

Suppose that $G^*(\overline{V_\alpha}-;\beta)<S(\overline{V_\alpha})< G^*(\overline{V_\alpha};\beta)$, let $b(\beta)=S(\overline{V_\alpha})$, then $a_1(\beta)<\overline{V_\alpha}<a_2(\beta)$. Since $S$ is decreasing on $[0,M]$, $G^*$ is decreasing and right-continuous on $[0,\overline{V_\alpha})$ and $[\overline{V_\alpha},M]$, it follows that
\begin{equation*}
\begin{aligned}
 &S(x)\ge S(\overline{V_\alpha})=b(\beta)\ge G^*(a_1;\beta)\ge G^*(x;\beta),\quad \forall x\in[a_1(\beta),\overline{V_\alpha}),\\
 &S(x)\le S(\overline{V_\alpha})=b(\beta)\le G^*(a_2-;\beta)\le G^*(x;\beta),\quad \forall x\in[\overline{V_\alpha},a_2(\beta)).     
\end{aligned}  
\end{equation*}
Since both $S$ and $G^*$ are greater than $S_0$, by the definition of $S(\cdot;b(\beta))$, we have $S(x;b(\beta))\ge S_0(x)$ and $S(\cdot;b(\beta))\in \tilde{\mathcal{S}}(S_0,\infty)$. 
For the case of $S(\overline{V_\alpha})\ge G^*(\overline{V_\alpha})$, let $b(\beta)=G^*(\overline{V_\alpha};\beta)$. If  $S(\overline{V_\alpha})\le G^*(\overline{V_\alpha}-;\beta)$, let $b(\beta)=G^*(\overline{V_\alpha}-;\beta)$. Since $O_{b}(\beta)$ is compact and $\mathscr{L}(S(x;b(\beta));\beta)$ is continuous in $b$ (by the Lebesgue Dominated Convergence Theorem), the optimal $b^*(\beta)$ is attainable. 

{Note that, by the proof of Proposition \ref{prop:beta>0}, under Assumption 
\ref{Assum:non-trivial}(\rmnum{2}), the integrand of 
$\mathcal{L}(S;\beta)$ is strictly concave in $S(x)$ for each $x\in[0,M]$  whenever $\beta>0$. Consequently, if  $S\in\tilde{\mathcal{S}}(S_0,\infty)$ is not of the form  \eqref{eq:modified-S}, the above construction yields some  $b(\beta)\in O_b(\beta)$ such that the Lebesgue measure of $\{x\in[0,M]:
S(x;b(\beta))\neq S(x)\}$ is positive. The strict pointwise improvement on this set, together with the weak pointwise improvement elsewhere, then implies
$\mathcal{L}(S(\cdot;b(\beta));\beta)>\mathcal{L}(S;\beta)$.

Now define
\begin{equation}
   \Psi(S):= \int_0^M\left([\varphi'(x)-\varphi'(0)]S(x)-\int_0^M\varphi''(y)(S_0(y)\wedge S(x))\,dy\right)\,dx.
\end{equation}
By the proof of Proposition \ref{prop:beta>0}, we get that any $\hat{S}\in \arg\sup_{S\in\tilde{\mathcal{S}}(S_0,\infty)}\mathscr{L}(S,0)$ must satisfy $\hat{S}(x)\ge G^*(x;0)$ for each $x\in[0,M]$. Thus, for $x\le \overline{V_\alpha}$, we have $\hat{S}(x)\ge \hat{S}(\overline{V_\alpha})\ge G^*(\overline{V_\alpha};0)$, which results in $\hat{S}(x)\ge G^*(x;0)\vee G^*(\overline{V_\alpha};0)$ on $[0,\overline{V_\alpha}]$. Consequently, $\hat{S}(x)\ge S(x;G^*(\overline{V_\alpha};0))$, where $S(x;G^*(\overline{V_\alpha};0))$ is defined by \eqref{eq:modified-S} with $b(0)=G^*(\overline{V_\alpha};0)$. Moreover, $\Psi(S)$ is increasing with respect to the pointwise
order on $\tilde{\mathcal S}(S_0,\infty)$; this is due to the fact that $\Psi(S_1)\ge \Psi(S_2)$ if $S_1(x)\ge S_2(x)$ for each $x\in[[0,M]$ when both $S_1, S_2\in\tilde{\mathcal{S}}(S_0,\infty)$. Thus, $S(\cdot;G^*(\overline{V_\alpha};0))$ has the smallest BW cost among all maximizers of $\mathscr{L}(S,0)$. Hence, $S(x;G^*(\overline{V_\alpha};0))$ solves Problem \eqref{Prob:inner} if $\Psi(S(x;G^*(\overline{V_\alpha};0)))\le \zeta$.

If $\Psi(S(x;G^*(\overline{V_\alpha};0)))> \zeta$, then under Lemma \ref{thm:sol-exist-inner} and the strong duality for the problem \eqref{Prob:inner}, the solution to the problem \eqref{Prob:inner} is characterized by the pair $(S^*,\beta^*)$, where $S^*\in\arg\sup_{S\in\tilde{\mathcal{S}}(S_0,\infty)}\mathscr{L}(S,\beta^*)$ and $\Psi(S^*)=\zeta$ (see Theorem 1 in Section 8.3 of \cite{luenberger1997optimization}). If $S^*$ is not of the form \eqref{eq:modified-S}, then by the first part of proof, there exist $b(\beta^*)\in O_b(\beta^*)$ and $b^*(\beta^*)\in B^*(\beta^*)$ such that
$$
\mathcal{L}(S(x;b^*(\beta^*));\beta^*)\ge \mathcal{L}(S(x;b(\beta^*));\beta^*)>\mathcal{L}(S^*;\beta^*),
$$
which contradicts with $S^*=\arg\sup_{S\in\tilde{\mathcal{S}}(S_0,\infty)}\mathscr{L}(S,\beta^*)$. Thus, $S^*$ is of the form \eqref{eq:modified-S}, and there exists a $b^*(\beta^*)\in B^*(\beta^*)$ such that $S^*=S(x;b^*(\beta^*))$ almost everywhere. The proof is now complete. \qed
}
\end{document}